\documentclass[11pt, a4paper]{article}

\usepackage[margin=2.5cm]{geometry}
\usepackage[T1]{fontenc}
\usepackage[sc]{mathpazo}
\usepackage[utf8]{inputenc} 
\usepackage{hyperref}              % customize internal links 
\usepackage{amsmath}
\usepackage{amssymb}
\usepackage{amsthm}
\usepackage{amsfonts}
\usepackage[mathscr]{eucal}        % lettres calligraphiques différentes de \mathcal
\usepackage{enumitem}              % customize items
\usepackage{multirow}
\usepackage{csquotes}
\usepackage{float}                 % forcer la position d'une image
\usepackage{tikz}
\usetikzlibrary{matrix, decorations.pathreplacing, calc, positioning, fit}
\usepackage[ruled,vlined,linesnumbered]{algorithm2e}
\usepackage{color, colortbl}
\usepackage[margin=0.5cm]{caption}
\usepackage[lofdepth,lotdepth]{subfig}
\usepackage{mathtools}
\usepackage{dsfont}                % pour le 'double 1' (fonction indicatrice) 
\hypersetup{colorlinks=false, linkcolor=black}%
\title{Lifted Projective Reed-Solomon Codes}
\date{}
\author{Julien \textsc{Lavauzelle}\\
  Laboratoire LIX, École Polytechnique, Inria \& CNRS UMR 7161\\
  Université Paris-Saclay\\
  \href{mailto:julien.lavauzelle@inria.fr}{julien.lavauzelle@inria.fr}
}
\newcommand\FF{\mathbb{F}}

\newcommand\NN{\mathbb{N}}
\newcommand\PP{\mathbb{P}}
\renewcommand\AA{\mathbb{A}}
\newcommand\EE{\mathbb{E}}

\newcommand\Deg{\mathrm{Deg}}
\newcommand\Lift{\mathrm{Lift}}
\newcommand\PLift{\mathrm{PLift}}
\newcommand\Diag{\mathrm{Diag}}

\newcommand\Hom{\mathrm{Hom}}
\newcommand\Emb{\mathrm{Emb}}
\newcommand\Poly{\mathrm{Poly}}
\newcommand\RM{\mathrm{RM}}
\newcommand\RS{\mathrm{RS}}
\newcommand\PRM{\mathrm{PRM}}
\newcommand\PRS{\mathrm{PRS}}
\newcommand\Aut{\mathrm{Aut}}
\newcommand\Perm{\mathrm{Perm}}
\newcommand\Aff{\mathrm{Aff}}
\newcommand\Proj{\mathrm{Proj}}
\newcommand\Iso{\mathrm{Iso}}
\newcommand\GL{\mathrm{GL}}
\newcommand\PDeg{\mathrm{PDeg}}
\newcommand\ADeg{\mathrm{ADeg}}

\DeclareMathOperator{\rank}{rank}
\DeclareMathOperator{\ev}{ev}

\newcommand\calC{\mathcal{C}}
\newcommand\calD{\mathcal{D}}

\newcommand\calF{\mathcal{F}}

\newcommand\calO{\mathcal{O}}
\newcommand\calP{\mathcal{P}}

\newcommand\calR{\mathcal{R}}

\newcommand\Corr{\mathsf{Corr}}
\newcommand\mydef{\coloneqq} %\newcommand\mydef{\vcentcolon:}
\theoremstyle{plain}
\newtheorem{theorem}{Theorem}
\newtheorem*{theorem*}{Theorem}
\newtheorem{proposition}[theorem]{Proposition}
\newtheorem*{proposition*}{Proposition}
\newtheorem{corollary}[theorem]{Corollary}
\newtheorem*{corollary*}{Corollary}
\newtheorem{lemma}[theorem]{Lemma}
\newtheorem*{lemma*}{Lemma}

\theoremstyle{definition}
\newtheorem{definition}[theorem]{Definition}
\newtheorem{remark}[theorem]{Remark}
\newtheorem{example}[theorem]{Example}
\newtheorem*{example*}{Example}
\begin{document}

\maketitle

\begin{abstract}
  Lifted Reed-Solomon codes, introduced by Guo, Kopparty and Sudan in 2013, are known as one of the few families of high-rate locally correctable codes. They are built through the evaluation over the affine space of multivariate polynomials whose restriction along any affine line can be interpolated as a low degree univariate polynomial.

  In this work, we give a formal definition of their analogues over projective spaces, and we study some of their parameters and features. Local correcting algorithms are first derived from the very nature of these codes, generalizing the well-known local correcting algorithms for Reed-Muller codes. We also prove that the lifting of both Reed-Solomon and projective Reed-Solomon codes are deeply linked through shortening and puncturing operations. It leads to recursive formulae on their dimension and their monomial bases. We finally emphasize the practicality of lifted projective Reed-Solomon codes by computing their information sets and by providing an implementation of the codes and their local correcting algorithms.
\end{abstract}

% \keywords{Lifted codes, local correcting algorithms, evaluation codes, projective space}
% \subclass{94B05, 11T71}

% \tableofcontents

% =============================================================================
\section{Introduction}

\paragraph{Motivation and previous works.} Locally decodable codes (LDC) and locally correctable codes (LCC) are codes equipped with a probabilistic algorithm which can efficiently decode or correct a single symbol of a noisy codeword, by querying only a few of its symbols. Low degree Reed-Muller codes define a well-known family of LDCs/LCCs with reasonable rate. Indeed, when restricted to an affine line, a sufficiently low-degree multivariate polynomial can be interpolated by a low-degree univariate polynomial. However, the rate $R$ of such Reed-Muller codes stays stuck below $1/2$. Multiplicity codes~\cite{KoppartySY14} were the first family of codes breaking the $R=1/2$ barrier for correcting a constant fraction of errors. The construction was based on a generalization of Reed-Muller codes which introduce multiplicities in the evaluation map. Shortly after the \emph{multiplicity codes} breakthrough, Guo, Kopparty and Sudan~\cite{GuoKS13} proposed another generalization of Reed-Muller codes and considered \emph{all} the multivariate polynomials (\emph{i.e.} not only the low-degree ones) which can be interpolated as low-degree univariate polynomials when restricted to a line. Surprisingly, it sometimes appears that much more polynomials satisfy this property than the low-degree ones lying in Reed-Muller codes. Resulting codes are named \emph{lifted Reed-Solomon codes}, and in this work, more shortly referred to as \emph{affine lifted codes}.

\paragraph{Organisation.} In this work, we show how to build analogues of these codes in projective spaces, that we call \emph{projective lifted codes}. Our construction relies on the notion of degree sets which also appears in~\cite{GuoKS13} and helps us to exhibit relations between affine and projective lifted codes. Section~\ref{sec:notation} introduce tools necessary to our construction. Affine and projective lifted codes are built in Section~\ref{sec:lifted-codes}, where we also prove main properties of projective lifted codes, notably their monomiality and the structure of their degree set. In Section~\ref{sec:local-correction} we present a family of local correcting algorithms for projective lifted codes, whose locality depends on the number of admissible errors on the queried line. Section~\ref{sec:relations} is devoted to the links between affine and projective lifted codes, through puncturing and shortening. Finally, we show miscellaneous properties of projective lifted codes in Section~\ref{sec:miscellaneous} which emphasize their explicitness and practicality: we present explicit information sets, we bound on their minimum distance and we prove their (quasi-)cyclicity under certain conditions.

We emphasize the practicality of our construction by presenting tables of parameters of projective lifted codes in Appendix~\ref{app:computation}. A basic implementation of affine and projective lifted codes in the open-source software SageMath~\cite{SageMath} is also made available\footnote{see \url{https://bitbucket.org/jlavauzelle/lifted_codes}}.

% =============================================================================
\section{Notation and preliminaries}
\label{sec:notation}

This section is devoted to introducing the algebraic background for the definition of affine and projective lifted codes.

% -----------------------------------------------------------------------------
\subsection{Geometry, polynomials and  evaluation maps}

We denote by $\FF_q$ the finite field with $q$ elements, and by $\FF_q^\times$ its non-zero elements. For $m \ge 1$, the \emph{affine space} of dimension $m$ is the set of $m$-tuples with coordinates in $\FF_q$, and is denoted $\AA^m$. We also define the \emph{projective space} of dimension $m$ as
\[
\PP^m \mydef (\AA^{m+1} \setminus \{ 0 \}) / \sim\,,
\]
where for $\mathbf{a}, \mathbf{b} \in \AA^{m+1} \setminus \{ 0 \}$, the relation $\sim$ is given by
\[
\mathbf{a}  \sim \mathbf{b} \iff \exists \lambda \in \FF_q^{\times}, \mathbf{a} = \lambda \mathbf{b} \,.
\]
A projective point will be denoted $\mathbf{a} = (a_0 : \dots : a_m) \in \PP^m$. It has $(q-1)$ different representatives, and we call \emph{standard representative} the only one such that  $\forall j < i, a_j = 0$ and $a_i = 1$. The projective space $\PP^m$ contains $\theta_{m,q} \mydef \frac{q^{m+1}-1}{q-1}$ distinct points.

The hyperplane at infinity $\Pi_{\infty} \mydef \{ \mathbf{a} \in \PP^m, a_0 = 0 \}$ is isomorphic to $\PP^{m-1}$, and the bijective map $(a_1, \dots, a_m) \mapsto (1 : a_1 : \dots : a_m)$ embeds $\AA^m$ into $\PP^m$. A \emph{projective line} is a $(q+1)$-subset of $\PP^m$ of the form
\[
L_{\mathbf{a}, \mathbf{b}} \mydef \{ x \mathbf{a} + y \mathbf{b}, (x : y) \in \PP^1 \} 
\]
for some distinct points $\mathbf{a}, \mathbf{b} \in \PP^m$. The line $L_{\mathbf{a}, \mathbf{b}}$ is the only one containing both $\mathbf{a}$ and $\mathbf{b}$, and there are exactly $\theta_{m-1,q} = |\PP^{m-1}| = \frac{q^m-1}{q-1}$ projective lines on which a given point $\mathbf{a} \in \PP^m$ lies.

\paragraph{Polynomials and degrees.} We denote by $\FF_q[\mathbf{X}] \mydef \FF_q[X_1, \dots, X_m]$ the ring of $m$-variate polynomials over $\FF_q$. Following the terminology given in~\cite{GuoKS13}, for $f = \sum_{\mathbf{d}} f_{\mathbf{d}} \mathbf{X^d} \in \FF_q[\mathbf{X}]$, the set $\{ \mathbf{d} \in \NN^m,  f_{\mathbf{d}} \ne 0\}$ is called the set of \emph{degrees} of $f$ and is denoted $\Deg(f)$. For a subset $D \subseteq \NN^m$, we denote by $\mathrm{Poly}(D)$ the vector space of polynomials generated by monomials $\mathbf{X^d}$ for $\mathbf{d} \in D$:
\[
\mathrm{Poly}(D) \mydef \langle \mathbf{X^d}, \mathbf{d} \in D \rangle \subseteq \FF_q[\mathbf{X}]\,.
\]
Some subsets $D$ are of particular interest. For instance, for $v \in \NN$,
\begin{itemize}
  \item the $1$-norm ball $B^m_1(v) \mydef \{ \mathbf{d} \in \NN^m, \sum_{i=1}^m d_i \le v \}$ generates the space $\FF_q[\mathbf{X}]_v$ of multivariate polynomials of total degree bounded by $v$,
  \item the $\infty$-norm ball $B^m_\infty(v) \mydef \{ \mathbf{d} \in \NN^m,  d_i \le v, \forall i =1,\dots,m \}$ generates the space of multivariate polynomials of partial degree bounded by $v$,
  \item the $1$-norm sphere  $S^{m+1}(v) \mydef \{ \mathbf{d} \in \NN^{m+1}, \sum_{i=0}^m d_i = v \}$ generates the space $\FF_q[\mathbf{X}]^H_v$ of homogeneous polynomials of degree $v$ (plus the zero polynomial).
\end{itemize}

We write $|\mathbf{d}| \mydef \sum_i d_i$ the \emph{weight} of a tuple of integers $\mathbf{d}$.

\paragraph{Evaluation of homogeneous polynomials on a projective point.} For any homogeneous polynomial $f \in \FF_q[\mathbf{X}]^H_v$, it is well-known that
\[
f(\lambda\mathbf{X}) = \lambda^vf(\mathbf{X})\,, \quad \forall \lambda \in \FF_q^{\times}\,.
\]
It means that different representatives of a fixed projective point may result to different evaluations by $f$. In order to remove any ambiguity, we adopt the following definition. Let $(a_0 : \dots : a_m)$ be the standard representation of a projective point $\mathbf{P} \in \PP^m$. Then we define the \emph{evaluation of $f$ at $\mathbf{P}$} as:
\[
\ev_{\mathbf{P}}(f) := f(a_0, \dots, a_m)\,.
\]
In other words, \emph{every projective point must be written in the unique standard representation} when evaluated by homogeneous polynomials.

Denote by $(\mathbf{P}_1, \dots, \mathbf{P}_{\theta_{m,q}})$ (resp. $(\mathbf{Q}_1, \dots, \mathbf{Q}_{q^m})$) an ordered list of all the projective (resp. affine) points. Thanks to the previous definition, the following evaluation map can be defined without ambiguity for all $v \ge 0$:
\[
\begin{array}{rclc}
  \ev_{\PP^m} : & \FF_q[\mathbf{X}]^H_v &\to &\FF_q^{\theta_{m,q}}\\
  & f & \mapsto & (\ev_{\mathbf{P}_1}(f), \dots, \ev_{\mathbf{P}_{\theta_{m,q}}}(f))
\end{array}
\]
Its affine analogue is:
\[
\begin{array}{rclc}
  \ev_{\AA^m} : & \FF_q[\mathbf{X}] &\to &\FF_q^{q^m}\\
  & f & \mapsto & (f(\mathbf{Q}_1), \dots, f(\mathbf{Q}_{q^m}))
\end{array}
\]
Clearly, $\ev_{\PP^m}$ and $\ev_{\AA^m}$ are $\FF_q$-linear maps. Since $x^q = x$ for all $x \in \FF_q$, we have
\[
\ker(\ev_{\AA^m}) = \langle X_i^q - X_i,\, \forall i=1,\dots,m  \rangle\,.
\]
Moreover, since $\ev_{\PP^m}$ evaluates homogeneous polynomials, for a fixed $v \in \NN$
\[
\ker(\ev_{\PP^m}) = \langle X_i^qX_j - X_iX_j^q,\, \forall i \ne j \in \{0,\dots,m\} \rangle \cap \FF_q[\mathbf{X}]^H_v\,.
\]
These properties are formally proved in~\cite{RenteriaT97}.

% -----------------------------------------------------------------------------
\subsection{Evaluation codes}

A common way to build linear codes is to evaluate polynomials over a list of points. In this subsection, we formally define the family of \emph{evaluation codes} we are studying. We also recall well-known examples of such codes, namely the Reed-Solomon and Reed-Muller codes, as well as their projective analogues.  

\begin{definition}[affine evaluation code]
  Let $\mathcal{F}$ be a linear subspace of $\FF_q[\mathbf{X}]$. The \emph{affine evaluation code} associated to  $\mathcal{F}$ is the $\FF_q$-linear code of length $n = |\AA^m| = q^m$ composed by the evaluation vectors of polynomials in $\mathcal{F}$:
\[
\ev_{\AA^m}(\mathcal{F}) = \{ \ev_{\AA^m}(f), f \in \mathcal{F} \} \subseteq \FF_q^{q^m}\,.
\]
The $n$-tuple $(\mathbf{Q}_1, \dots, \mathbf{Q}_n)$ of evaluation points is called the \emph{support} of the code.
\end{definition}

\begin{definition}[projective evaluation code]
  Let $\mathcal{F}$ be a linear subspace of $\FF_q[\mathbf{X}]_v^H$. The \emph{projective evaluation code} associated to $\mathcal{F}$ is the $\FF_q$-linear code of length $n = |\PP^m| = \theta_{m,q}$ composed by the evaluation vectors of polynomials in $\mathcal{F}$:
\[
\ev_{\PP^m}(\mathcal{F}) = \{ \ev_{\PP^m}(f), f \in \mathcal{F} \} \subseteq  \FF_q^{\theta_{m,q}} \,.
\]
Once again, the $n$-tuple $(\mathbf{P}_1, \dots, \mathbf{P}_n)$ of evaluation points is called the \emph{support} of the code.
\end{definition}

We point out a specific class of evaluation codes which is generated by evaluation vectors of monomials.

\begin{definition}[monomial code]
  An affine evaluation code $\calC = \ev_{\AA^m}(\mathcal{F})$ (resp. a projective evaluation code $\calC = \ev_{\PP^m}(\mathcal{F})$) is said \emph{monomial} if $\mathcal{F} = \mathrm{Poly}(D)$ for some $D \subseteq \NN^m$ (resp. $D \subseteq S^{m+1}(v)$).
\end{definition}

As we will see later, monomial codes turn out to be very convenient to describe with their set of degrees $D$.

\paragraph{Reed-Solomon and Reed-Muller codes.}

\begin{definition}[Reed-Solomon code]
  Let $0 \le k \le q-1$. The vector space of evaluation vectors of polynomials of degree $\le k$ over $\FF_q$ is called the (full-length) \emph{Reed-Solomon code}:
  \[
  \RS_q(k) \mydef \{ \ev_{\AA^1}(f), f \in \FF_q[X]_k \}\,.
  \]
  and has dimension $k+1$ over $\FF_q$.
\end{definition}

\begin{definition}[Reed-Muller code]
   Let $0 \le d \le m(q-1)$. The (generalized) \emph{Reed-Muller code} of order $m$ and degree $d$ over $\FF_q$ is the subspace of $\FF_q^{q^m}$ consisting in evaluation vectors of $m$-variate polynomials over $\FF_q$ of total degree $\le d$:
  \[
  \mathrm{RM}_q(m, d) \mydef \{ \ev_{\AA^m}(f), f \in \FF_q[\mathbf{X}]_d \}\,.
  \]
  For $d>0$, the dimension of $\mathrm{RM}_q(m, d)$ is given by~\cite{AssmusK92}:
  \[
  \dim \mathrm{RM}_q(m, d) = \sum_{i=0}^d \sum_{j=0}^m (-1)^j \binom{m}{j}\binom{i-jq+m-1}{i-jq}\,,
  \]
  and simplifies to $\binom{m+d}{m}$ for $d \le q-1$.
\end{definition}

Reed-Muller codes generalize Reed-Solomon codes, in the sense that $\mathrm{RM}_q(1,k) = \RS_q(k)$.

\paragraph{Projective Reed-Solomon and Reed-Muller codes.} Previous codes can be naturally adapted to the context of projective spaces.

\begin{definition}[Projective Reed-Solomon code]
  Let $0\le k \le q$. The \emph{projective Reed-Solomon code} of dimension $k+1$ over $\FF_q$ is the linear code of length $q+1 = |\PP^1|$ consisting of the evaluation of bivariate homogeneous polynomials of degree $k$ over $\FF_q$:
  \[
  \PRS_q(k) \mydef \{ \ev_{\PP^1}(f), f \in \FF_q[X,Y]^H_k \}\,.
  \]
\end{definition}

Projective Reed-Solomon codes are also called \emph{extended}, or \emph{doubly-extended} Reed-Solomon codes. Similarly, Reed-Muller codes have a projective analogue, defined as follows~\cite{Lachaud86, Lachaud90, Sorensen91}:

\begin{definition}[Projective Reed-Muller code]
  Let $1 \le v \le m(q-1)$. The \emph{projective Reed-Muller code} of order $m$ and degree $v$ over $\FF_q$ is the linear code of length $|\PP^m| = (q^{m+1}-1)/(q-1)$ consisting of the evaluation of $(m+1)$-variate homogeneous polynomials over $\FF_q$ of degree $v$:
  \[
  \mathrm{PRM}_q(m, v) \mydef \{ \ev_{\PP^m}(f), f \in \FF_q[\mathbf{X}]^H_v \}\,.
  \]
  The dimension of $\mathrm{PRM}_q(m, v)$ is (see~\cite{Sorensen91}):
  \[
  \dim \mathrm{PRM}_q(m, v) = \sum_{t \in I_v} \left( \sum_{j=0}^{m+1} (-1)^j \binom{m+1}{j} \binom{t-jq+m}{t-jq} \right)\,,
  \]
where $I_v = \{ t \in [1,v],\, t \equiv v \!\mod q-1\}$. For $v \le q-1$, it simplifies to $\dim \mathrm{PRM}_q(m, v) = \binom{m+v}{v}$.
\end{definition}

Once again, by definition we have $\PRS_q(k) = \mathrm{PRM}_q(1, k)$ for every $1 \le k \le q-1$.

% -----------------------------------------------------------------------------
\subsection{Reduced degree sets}

In the previous subsection, we have seen that well-known families of linear codes are defined as the image of subspaces of polynomials by evaluation maps. For coding theoretic reasons (\emph{e.g.} giving the dimension of the code, or computing a basis), it is interesting to find sets $D \subseteq S^{m+1}(v)$ (resp. $D \subseteq B_1^m(v)$) such that the map $\ev_{\PP^m}$  (resp.  $\ev_{\AA^m}$) is injective over $\mathrm{Poly}(D)$. So let us define such sets.

First of all, we introduce specific tuples.
\begin{definition}[$A$ and $P$-reduced tuples]~
  \begin{enumerate}
    \item A tuple $\mathbf{d} \in \NN^m$ is \emph{$A$-reduced} if ${\bf d}$ lies in $B^m_\infty(q-1)$.
    \item A tuple ${\bf d} = (d_0, \dots, d_m) \in \NN^{m+1}$ is \emph{$P$-reduced} if, for all $0 \le i \le m$:
  \[
  d_i \ge q \quad \Rightarrow \quad \left\{ \begin{array}{ll} d_j = 0 &\forall j < i\,,\\d_j \le q-1 &\forall j > i\,. \end{array}\right.
  \]
\end{enumerate}
  We see that any $A$-reduced tuple is also $P$-reduced. We also say that a set $D$ of tuples is \emph{$A$-reduced} (resp. \emph{$P$-reduced}) if every tuple it contains is $A$-reduced (resp. $P$-reduced).
\end{definition}

Denote by $(\mathbf{e}_1,\dots,\mathbf{e}_m)$ the canonical basis of $\NN^m$. Let $\mathbf{d} \in \NN^m$, $1 \le i< j \le m$, and assume that $d_j \ge q$ and $d_i \ge 1$. For such $\mathbf{d}$ (and only for such $\mathbf{d}$), we define $\rho_j(\mathbf{d}) \mydef \mathbf{d} - (q-1)\mathbf{e}_j$ and $\tau_{ij}(\mathbf{d}) \mydef \mathbf{d} + (q-1)(\mathbf{e}_i-\mathbf{e_j})$. Remark that $|\tau_{ij}({\bf d})| = |{\bf d}|$.

\begin{remark}
  \label{rem:reduction}
  Let $\mathbf{d} \in \NN^m$ or $\NN^{m+1}$ depending on the context (affine or projective). Then, 
  \begin{itemize}
  \item we have $\ev_{\AA^m}(\mathbf{X}^{\rho_j(\mathbf{d})}) = \ev_{\AA^m}(\mathbf{X^{\mathbf{d}}})$ and $\ev_{\PP^m}(\mathbf{X}^{\tau_{ij}(\mathbf{d})}) = \ev_{\PP^m}(\mathbf{X^{\mathbf{d}}})$;
  \item as long as they are defined, $\rho_j \circ \rho_\ell = \rho_\ell \circ \rho_j$ and $\tau_{ij} \circ \tau_{k\ell} = \tau_{k\ell} \circ \tau_{ij}$;
  \item if no $\tau_{ij}$ can be applied to $\mathbf{d}$, then $\mathbf{d}$ is $P$-reduced;
  \item if no $\rho_j$ can be applied to $\mathbf{d}$, then $\mathbf{d}$ is $A$-reduced;
  \item if we keep applying  to some tuple $\mathbf{d}$ the maps $\tau_{ij}$, for $0 \le i< j \le m$, until we cannot apply any of them, then we obtain a $P$-reduced tuple;
  \item if we keep applying to some tuple $\mathbf{d}$ the maps $\rho_j$, for $1 \le j \le m$, until we cannot apply any of them, then we obtain an $A$-reduced tuple.
  \end{itemize}
\end{remark}

\begin{definition}
  Let $\mathbf{d} \in \NN^m$. The \emph{$A$-reduction} of $\mathbf{d}$ is the tuple $\underline{\mathbf{d}} \in \NN^m$ which is obtained by applying iteratively $\rho_j$ (for $1 \le j \le m$) until the result lies in $B^m_\infty(q-1)$. It satisfies $\ev_{\AA^m}(\mathbf{X^{\underline{\mathbf{d}}}}) = \ev_{\AA^m}(\mathbf{X^{\mathbf{d}}})$. The $A$-reduction of $D \subseteq \NN^{m+1}$, denoted $\underline{D}$, consists in the $A$-reduction of the tuples in $D$.
\end{definition}

\begin{definition}
  Let $\mathbf{d} \in S^{m+1}(v)$ for some $v>0$. The \emph{$P$-reduction} of $\mathbf{d}$ is the tuple $\overline{\mathbf{d}} \in S^{m+1}(v)$ which is obtained by applying iteratively $\tau_{ij}$ (for $0 \le i < j \le m$) until the result is $P$-reduced. It satisfies $\ev_{\PP^m}(\mathbf{X^{\overline{\mathbf{d}}}}) = \ev_{\PP^m}(\mathbf{X^{\mathbf{d}}})$. The $P$-reduction of $D \subseteq S^{m+1}(v)$, denoted $\overline{D}$, consists in the $P$-reduction of the tuples in $D$.
\end{definition}

$A$- and $P$-reduction are defined in order to make the evaluation maps $\ev_{\AA^m}$ and $\ev_{\PP^m}$ injective over polynomial spaces of the form $\Poly(D)$, where $D$ is $A$- or $P$-reduced. Next lemma details these properties.

\begin{lemma}
  Let $m \ge 1$ and $v \in \NN$. The following properties hold:
  \begin{enumerate}
  \item If $D \subseteq \NN^m$ is $A$-reduced, then the map $\ev_{\AA^m}$ is injective over $\mathrm{Poly}(D)$.
  \item If $D \subseteq S^{m+1}(v)$ is $P$-reduced, then the map $\ev_{\PP^m}$ is injective over $\mathrm{Poly}(D)$.
  \item For every $D \subseteq \NN^m$, the $A$-reduction $\underline{D}$ of $D$ is the unique $A$-reduced subset of $\NN^m$ satisfying
    \[
    \ev_{\AA^m}(\mathrm{Poly}(D)) = \ev_{\AA^m}(\mathrm{Poly}(\underline{D}))\,.
    \]
  \item For every $D \subseteq S^{m+1}(v)$, the $P$-reduction $\overline{D}$ of $D$ is the unique $P$-reduced subset of $S^{m+1}(v)$ satisfying
    \[
    \ev_{\PP^m}(\mathrm{Poly}(D)) = \ev_{\PP^m}(\mathrm{Poly}(\overline{D}))\,.
    \]
  \end{enumerate}
\end{lemma}

\begin{proof}~
  \begin{enumerate}
  \item By definition, if $D$ is $A$-reduced, then $D$ is a subset of $B_\infty^m(q-1)$. Since $\ker(\ev_{\AA^m}) = \langle X_i^q - X_i \rangle_{1 \le i \le m}$, we can see that $\mathrm{Poly}(D) \cap \ker(\ev_{\AA^m}) = \{ 0 \}$.
  \item We proceed by induction on $m$. Recall that $\ker(\ev_{\PP^m}) = \langle X_i^qX_j - X_iX_j^q \rangle_{0 \le i < j \le m}  \cap \FF_q[\mathbf{X}]^H_v$.
\begin{itemize}
  \item For $m=1$ and $v \in \NN$, let $D$ be a $P$-reduced subset of $S^2(v)$. If $v \le q$, it is clear that $\ker(\ev_{\PP^1}) \cap \Poly(D) = \{ 0 \}$. So assume $v>q$ and let $f(X,Y) \in \mathrm{Poly}(D) \cap \ker(\ev_{\PP^1})$. Since $D$ is $P$-reduced, we can write
\[
f(X,Y) \mydef  f_v Y^v + \sum_{i=0}^{q-1} f_i X^{v-i}Y^i \in \mathrm{Poly}(D)\,.
\]
Then, we see that $f_v = f(0,1) = 0$, hence $g(Y) \mydef f(1, Y) = \sum_{i=0}^{q-1} f_i Y^i$ lies in $\mathrm{Poly}(D')$, for some set $D' \subseteq B^1_\infty(q-1)$. Moreover $f \in \ker(\ev_{\PP^1})$ implies $g \in \ker(\ev_{\AA^1})$. Hence, the first point of this Lemma (applied to $D'$ which is $A$-reduced) shows that $g=0$, and $f=0$ follows.
  \item Let $m > 1$ and $v \in \NN$. The proof works similarly. Let $D$ be a $P$-reduced subset of $S^{m+1}(v)$, and let
    \[
    f({\bf X}) \mydef f_0(X_1,\dots,X_m) + X_0 f_1(X_0,X_1,\dots,X_m) \in \mathrm{Poly}(D) \cap \ker(\ev_{\PP^m})\,.
    \]
    Since $f_0$ does not depend on $X_0$, we can see that $f_0 \in \ker(\ev_{\PP^{m-1}})$ and $f_0 \in \mathrm{Poly}(D_0)$ where $D_0 \subset S^m(v)$. Besides, $D_0$ is $P$-reduced as a subset of $D$. Therefore, by induction $f_0 = 0$, and $f = X_0 f_1(X_0,X_1,\dots,X_m)$ follows. Let us define $g \mydef f(1,X_1,\dots,X_m)$; we see that $g \in \ker(\ev_{\AA^m})$ and $g \in \mathrm{Poly}(D')$ where $D' \subseteq B^m_\infty(q-1)$ since $D$ is $P$-reduced and every degree tuple in $D'$ comes from a tuple ${\bf d} \in D$ such that $d_0 \ne 0$. Thanks to the first point of the lemma, it follows that $g = 0$. Therefore, $f = \alpha X_0^v$ with $\alpha \in \FF_q$, which necessarily implies $f = 0$ (evaluate it at $(1:0:\dots:0)$).
\end{itemize}
  \item Since $\ev_{\AA^m}(\mathbf{X^{\underline{\mathbf{d}}}}) = \ev_{\AA^m}(\mathbf{X^{\mathbf{d}}})$ for every ${\bf d} \in \NN^m$, we have $\ev_{\AA^m}(\mathrm{Poly}(\underline{D})) = \ev_{\AA^m}(\mathrm{Poly}(D))$. Uniqueness comes from the injectivity of $\ev_{\AA^m}$.
  \item Same argument.
  \end{enumerate}
\end{proof}

\begin{definition}[Degree set]
 Let $\calC = \ev_{\AA^m}(\mathrm{Poly}(D))$ be an affine (resp. let $\calC = \ev_{\PP^m}(\mathrm{Poly}(D))$ be a projective) monomial code. Its \emph{degree set} is the unique $A$-reduction (resp. $P$-reduction) of $D$, and is denoted $\mathrm{Deg}(\calC)$. 
\end{definition}

By definition, if $\calC$ is monomial, then we have $\calC = \ev(\mathrm{Poly}(\Deg(\calC)))$ where $\ev \in \{\ev_{\AA^m}, \ev_{\PP^m}\}$ depending on the context. Moreover, since $\ev$ is injective over $\mathrm{Poly}(\Deg(\calC))$, it also holds that:
  \[
  \dim \calC = |\mathrm{Deg}(\calC)|\,. 
  \]

\begin{example}
Reed-Solomon and Reed-Muller codes, as well as their projective analogues, are monomial codes. Table~\ref{tab:degree-sets-RS-RM} presents their degree sets.

\begin{table}[h!]
  \centering
  \setlength\extrarowheight{8pt}
  \begin{tabular}{l|l}
    \textbf{Code} & \textbf{Degree set} \\
    \hline
    Reed-Solomon code $\RS_q(k)$ &  $B^1_1(k) = \{ 0, 1, \dots, k \}$ \\
    Reed-Muller code $\mathrm{RM}_q(m, k)$ & $\underline{B^m_1(k)} = \left\{ \underline{\mathbf{e}} \mid \mathbf{e} \in \NN^m, |\mathbf{e}| \le k \right\} $\\
    projective Reed-Solomon code $\PRS_q(k)$ & $S^2_1(k) = \{ (k, 0), (k-1, 1), \dots, (0, k) \} $ \\
    projective Reed-Muller code $\mathrm{PRM}_q(m, k)$ & $\overline{S^{m+1}_1(k)} = \{ \overline{\mathbf{d}} \mid \mathbf{d} \in \NN^{m+1}, |\mathbf{d}| = k \}$    
  \end{tabular}
  \caption{\label{tab:degree-sets-RS-RM}Degree sets of classical monomial codes.}
\end{table}
\end{example}

% -----------------------------------------------------------------------------
\subsection{Permutations, automorphisms}
\label{subsec:automorphims}

Generally, a linear code $\calC$ is a linear subspace of $\FF_q^X$ for some finite set $X$. Any permutation $\sigma$ of $X$ induces a permutation of the coordinates of vectors $c \in \FF_q^X$ given by:
\[
\sigma(c) = (c_{\sigma^{-1}(x)})_{x \in X} \in \FF_q^X
\]
Denote by $\frak{S}(X)$ the group of permutations of $X$. The subset of permutations of $X$ which let $\calC$ invariant is a subgroup of $\frak{S}(X)$, called the \emph{permutation group} of $\calC$, and denoted by $\mathrm{Perm}(\calC)$.

Let $\Iso(X)$ be the semi-direct product $\big(\FF_q^{\times}\big)^X \rtimes \frak{S}(X)$. Any $(w, \sigma) \in \Iso(X)$ acts on $c \in \FF_q^X$ by:
\[
(w, \sigma) \cdot c = w \star \sigma(c)\,,
\]
where $\star$ denotes the component-wise product between tuples: ${\bf a} \star {\bf b} = (a_1b_1, \dots, a_m,b_m)$. 
If $w = (1,\dots,1)$, we simply write $\sigma \in \Iso(X)$. 

The subgroup of $\Iso(X)$ letting $\calC$ invariant is called the \emph{automorphism group of $\calC$}, and is denoted by $\Aut(\calC)$. Of course, $\Perm(\calC) \subseteq \Aut(\calC)$.

Let us finally denote by $\mathrm{GL}_m(\FF_q)$ the group of $m \times m$ invertible matrices over $\FF_q$. Using the canonical basis, these matrices represent linear automorphisms $\AA^m \to \AA^m$.

\paragraph{Affine evaluation codes.}  In the case $X = \AA^m$, let us define the affine transformations $T_{M, \mathbf{b}} : \AA^m \to \AA^m$ by $\mathbf{x} \mapsto M\mathbf{x} + \mathbf{b}$, for every $M \in \mathrm{GL}_m(\FF_q)$ and $\mathbf{b} \in \FF_q^m$. Each $T_{M, \mathbf{b}}$ is a permutation of $\AA^m$. Denote by $\mathrm{Aff}(\FF_q, m)$ the group of such transformations:
\[
\Aff(\FF_q, m) = \{ T_{M, \mathbf{b}} \mid (M, \mathbf{b}) \in \GL_m(\FF_q) \times \FF_q^m \}\,.
\]
In Appendix~\ref{app:automorphism-groups-codes} we prove that $\mathrm{Aff}(\FF_q, m) \subseteq \mathrm{Perm}(\RM_q(m, k))$ for every $0 \le k \le m(q-1)$.

\paragraph{Projective evaluation codes.} In the case $X = \PP^m$, let $M \in \GL_{m+1}(\FF_q)$. Then $\mathbf{x} \mapsto M \mathbf{x}$ induces a permutation of $\PP^m$, but does not necessarily preserve the standard representation of projective points. Still, there exists $\lambda_{M,\mathbf{x}} \in \FF_q^{\times}$ such that the standard representative of $M\mathbf{x}$ is $\lambda_{M,\mathbf{x}} M \mathbf{x}$. For every $f \in \FF_q[\mathbf{X}]^H_v$, we then have:
\[
\ev_{M\mathbf{x}}(f) = f(\lambda_{M,\mathbf{x}}M \mathbf{x}) = (\lambda_{M,\mathbf{x}})^v\, f(M \mathbf{x}) = (\lambda_{M,\mathbf{x}})^v\, \ev_{\mathbf{x}}(f \circ M)\,,
\]
and we see that $(\lambda_{M,\mathbf{x}})^v$ does not depend on $f$ (only on its total degree). Let us denote by $w_M^v \mydef ((\lambda_{M,\mathbf{x}})^v : \mathbf{x} \in \PP^m)$, and by $\sigma_M$ the permutation of $\PP^m$ induced by $M$. Then, we have:
\[
\ev_{\PP^m}(f \circ M) = w_M^v \, \star \, \sigma_{M^{-1}}(\ev_{\PP^m}(f))\,.
\]

Denote by $\Proj(\FF_q, m) \mydef \{ (w_M^v, \sigma_{M^{-1}}) \in \Iso(\PP^m), M \in \GL_{m+1}(\FF_q) \}$ in the context of evaluating homogeneous polynomials of degree $v$. In Appendix~\ref{app:automorphism-groups-codes} we prove that $\Proj(\FF_q, m) \subseteq \Aut(\PRM_q(m, v))$ for every $1 \le v \le m(q-1)$.

% -----------------------------------------------------------------------------
\subsection{Embedding maps}

Here we define maps embedding lines into higher dimensional spaces. For $U, V$ two $\FF_q$-linear spaces, we denote by $\mathrm{Hom}(U, V)$ the set of linear maps $U \to V$. Let $\mathrm{Emb}_\PP(m)$ be the set of full-rank (\emph{i.e.} injective) linear maps from $\FF_q^2$ to $\FF_q^{m+1}$:
\[
\mathrm{Emb}_\PP(m) \mydef \{ L \in \mathrm{Hom}(\FF_q^2, \FF_q^{m+1}),\, \rank L = 2 \}\,.
\]
Each $L \in \mathrm{Emb}_\PP(m)$ induces a \emph{projective embedding} $\PP^1 \to \PP^m$ sending $(x : y) \mapsto L(x,y)$. One can easily check that this map is well defined over projective spaces. Moreover, the set $\{ L(\PP^1), L \in \mathrm{Emb}_\PP(m) \}$ describes all the projective lines of $\PP^m$, though a projective line is obviously associated to many maps $L$ in $\mathrm{Emb}_\PP(m)$.

Similarly, the set
\[
\mathrm{Emb}_\AA(m) \mydef \{ L^* = (L_1,\dots,L_m) \in \mathrm{Hom}(\FF_q^2, \FF_q^m) \mid L = (L_0, \dots, L_m) \in \mathrm{Emb}_\PP(m)  \}
\]
defines \emph{affine embeddings} $\AA^1 \to \AA^m$ by $t \mapsto L^*(1,t)$. The set $\{ L^*(1,\AA^1), L^* \in \mathrm{Emb}_\AA(m) \}$ defines the set of affine lines of $\AA^m$.

\begin{remark}
  Elements of $\mathrm{Emb}_\PP(m)$ and $\mathrm{Emb}_\AA(m)$ will sometimes be seen as $(m+1) \times 2$ or $m \times 2$ matrices over $\FF_q$. Besides, for convenience and when the context is clear, we will improperly write $L^*(t)$ instead of $L^*(1,t)$. By using this notation, we want to emphasize that, for every $f \in \FF_q[{\bf X}]$ and every $L^* \in \mathrm{Emb}_\AA(m)$, the map $t \mapsto f(L^*(1,t))$ can be interpolated as a univariate polynomial denoted $f \circ L^* \in \FF_q[T]$.
\end{remark}

\begin{remark}
  \label{rem:subword}
  For local correction purposes (see section~\ref{sec:local-correction}), it is important to notice the following points.
\begin{enumerate}
\item In the affine setting, for every $L^* \in \Emb_\AA(m)$ and $f \in \FF_q[\mathbf{X}]$, the word $\ev_{\AA^1}(f \circ L^*)$ is a subword of $\ev_{\AA^m}(f)$, and can be read at indices $L^*(t)$ for $t \in \AA^1$.
\item In the projective setting, $\ev_{\PP^1}(f \circ L)$ is not necessary a subword of $\ev_{\PP^m}(f)$, since nothing asserts that $L$ preserves the standard representation of projective points, similarly to the discussion in Subsection~\ref{subsec:automorphims}. We solve this issue in a very similar manner. Let $\mathbf{x} \in \PP^1$  and $L \in \mathrm{Emb}_\PP(m)$. We know there exists $\lambda_{L,\mathbf{x}} \in \FF_q^{\times}$ such that the standard representative ${\bf P}$ of $L(\mathbf{x})$ is $\lambda_{L,\mathbf{x}} L(\mathbf{x})  \in \PP^m$. Then it holds:
\[
\forall f \in \FF_q[\mathbf{X}]_v^H,\,\,  \ev_{\mathbf{P}}(f) = f(\lambda_{L,\mathbf{x}}L(\mathbf{x})) = (\lambda_{L,\mathbf{x}})^v (f \circ L)(\mathbf{x})\,.
\]
Therefore, let us define $w^v_L = ((\lambda_{L,\mathbf{x}})^v : \mathbf{x} \in \PP^1)$. Then $w^v_L \star \ev_{\PP^1}(f \circ L)$ is a subword of $\ev_{\PP^m}(f)$, and can be read at indices $L(\mathbf{x}) \in \PP^m$ for $\mathbf{x} \in \PP^1$.
\end{enumerate} 
\end{remark}

\begin{example}
  Let us fix an ordered list of points in $\PP^1(\FF_3)$ and $\PP^2(\FF_3)$:
  \[
  \begin{array}{rl}
  \PP^1(\FF_3) = &\big( (1:1), (1:2), (1:0), (0:1) \big)\\
  \PP^2(\FF_3) = &\big( (1:1:1), (1:1:2), (1:1:0), (1:2:1), (1:2:2), (1:2:0),\\
  &\quad(1:0:1), (1:0:2), (1:0:0), (0:1:1), (0:1:2), (0:1:0), (0:0:1) \big)
\end{array}
  \]
  Let $f = X_1 \in \FF_3[X_0, X_1, X_2]^H_1$ and $L \in \Emb_\PP(2)$ represented by the matrix:
  \[
  L = \begin{pmatrix} 1 & 1 \\ 0 & 1 \\ 1 & 0 \end{pmatrix}
  \]
  Denote by $c = \ev_{\PP^2}(f) = (1, 1, 1, 2, 2, 2, 0, 0, 0, 1, 1, 1, 0) \in \FF_3^{13}$. On the one hand we have
  \begin{equation}
  \label{eq:homogenizing}
  \begin{aligned}
  L(\PP^1) &= ( (2:1:1), (0:2:1), (1:0:1), (1:1:0) )\\
           &= ( (1:2:2), (0:1:2), (1:0:1), (1:1:0) )\,,
  \end{aligned}
  \end{equation}
 hence $\ev_{\PP^2}(f)_{|L(\PP^1)} = (c_5, c_{11}, c_7, c_3) =  (2, 1, 0, 1)$. On the other hand $(f \circ L)(S, T) = T \in \FF_3[S, T]^H_1$, and we get $\ev_{\PP^1}(f \circ L) = (1, 2, 0, 1)$. Clearly $\ev_{\PP^2}(f)_{|L(\PP^1)} \ne \ev_{\PP^1}(f \circ L)$.

 Nevertheless, $w_L^1$ can be obtained through the homogenizing made in~\eqref{eq:homogenizing}:
 \[
 w_L^1 = (2^1, 2^1, 1^1, 1^1)\,.
 \]
 Therefore it gives:
 \[
 w_L^1 \star \ev_{\PP^1}(f \circ L) = (2, 1, 0, 1) = \ev_{\PP^2}(f)_{|L(\PP^1)}\,.
 \]
\end{example}

% =============================================================================
\section{Affine and projective lifted Reed-Solomon codes}
\label{sec:lifted-codes}

Before introducing our construction, we recall the definition of affine lifted codes given by Guo, Kopparty and Sudan~\cite{GuoKS13}. Notice that we restrict our study to the lifting of (projective) Reed-Solomon codes, but we believe that our construction can be extrapolated to the lifting of (projective) Reed-Muller codes. Besides, our formalism is slightly different from the paper of~\cite{GuoKS13}, since their notion of restriction  $f_{|L}$ of a polynomial $f$ along a line $L$ is somewhat ambiguous.

% -----------------------------------------------------------------------------
\subsection{Affine lifted codes}
\label{subsec:affine-lifted-codes}

We first need to introduce a few notation.
\begin{itemize}
\item Let $a = \sum a^{(i)} p^i$ be the $p$-adic decomposition of a non-negative integer $a$. We define a partial order $\le_p$ on integers by:
  \begin{equation}
    \label{eq:def_lep}
  a \le_p b \iff a^{(i)} \le b^{(i)},\, \forall i\,.
  \end{equation}
  The relation $\le_p$ can be naturally extended to $m$-tuples by $\mathbf{a} \le_p \mathbf{b} \iff \forall j,\, a_j \le_p b_j$. 
\item We also extend binomial coefficients to $m$-tuples of integers by $\binom{\mathbf{a}}{\mathbf{b}} \mydef \prod_{i=1}^m \binom{a_i}{b_i}$.
\end{itemize}

We also recall that, for $f \in \FF_q[{\bf X}] $ and $L \in \Emb_{\AA}(m)$, the notation $f \circ L$ represents the univariate polynomial $f(L(1,T))$.

\begin{definition}[Affine lifting of Reed-Solomon codes~\cite{GuoKS13}]
  Let $0 \le k \le q-2$ and $m \ge 1$. The \emph{affine lifting} of order $m$ of the Reed-Solomon code $\RS_q(k)$ is
  \[
  \Lift_q(m, k) \mydef \{ \ev_{\AA^m}(f) \mid f \in \FF_q[\mathbf{X}], \forall L \in \Emb_\AA(m), \ev_{\AA^1}(f \circ L) \in \RS_q(k) \}\,.
  \]
\end{definition}

The codes $\Lift_q(m, k)$ will shortly be called \emph{affine lifted codes}. In~\cite{GuoKS13} it is also proved that every affine lifted code $\Lift_q(m, k)$ is monomial and satisfies
  \begin{equation}\label{eq:degree-set-lift}
  \Lift_q(m, k) = \left\langle \ev_{\AA^m}(\mathbf{X^d}) \mid \mathbf{d} \in B_\infty^m(q-1), \forall \mathbf{e} \le_p \mathbf{d}, \overline{|\mathbf{e}|} \le k \right\rangle,
  \end{equation}
where $p = {\rm char}(\FF_q)$ and $|{\bf e}| = \sum_i e_i$. Note that monomiality of affine lifted code follows from their affine-invariance, by using a result of Kaufman and Sudan~\cite{KaufmanS08}.

A careful observation of their degree sets shows that $\Lift_q(m, k)$ fits between two projective Reed-Muller codes:
\begin{equation}
  \label{eq:lift-inclusion}
  \mathrm{RM}_q(m, k) \subseteq \Lift_q(m, k) \subseteq \mathrm{RM}_q(m, k + (m-1)(q-1))\,.
\end{equation}

The main interest of affine lifted codes appears for some values of $q$ and $k$ (essentially $q$ non-prime and $k$ close to $q$), for which the first inclusion is proper. Indeed, Kaufman and Ron give in~\cite{KaufmanR06} arguments that shows that affine lifted codes are Reed-Muller codes as long as $k < q - \frac{q}{p}$ (where $p$ is the characteristic of the field).

In the $k \ge q - \frac{q}{p}$ setting, some families of affine lifted codes give rise to a family of high-rate locally decodable and correctable codes, while Reed-Muller codes have rate bounded by $1/m!$. More specifically, the following theorem is proved in~\cite{GuoKS13} (we report the formal definition of locally correctable codes to Section~\ref{sec:local-correction}):

\begin{theorem}[High rate lifted codes, \cite{GuoKS13}]
  Let $0 < \rho, \gamma < 1$ and $n_0 \ge 1$. Define $m \mydef \lceil 1 / \delta \rceil$, $q \mydef 2^s \ge n_0^\gamma$, $b \mydef 1 + \lceil \log m \rceil$ and $c \mydef \lceil b 2^{bm} \log(1/\rho) \rceil$. Finally, let $k \mydef (1-2^{-c})q$. Then the code $\Lift_q(m, k)$ has length $n \ge n_0$, rate $R \ge 1 - \rho$, and is locally correctable with locality $\ell = n^\gamma$ for a $\delta = 2^{-c}/6$ fraction of errors.
\end{theorem}

However, for generic parameters $m, k, q$, exact formulae for the dimension of affine lifted codes are hard to produce. We give some concrete values in Appendix~\ref{app:computation}.

% -----------------------------------------------------------------------------
\subsection{Projective lifted codes}

In this section, we aim at defining the projective analogues of the lifted Reed-Solomon codes introduced by Guo \emph{et al.}~\cite{GuoKS13}. A way to build an evaluation code over a projective space is to evaluate homogeneous polynomials of fixed degree $v$, as it is done for projective Reed-Muller codes. It raises the problem of determining a meaningful value of $v$ we could use to define projective lifted codes. Equation~\eqref{eq:lift-inclusion} suggests to set $v = v_{m,k} \mydef k + (m-1)(q-1)$.

\begin{definition}
  Let $1 \le k \le q-1$, $m \ge 1$ and $v = v_{m,k} = k + (m-1)(q-1)$. The projective lifting of order $m$ of the projective Reed-Solomon code $\PRS_q(k)$ is
  \[
  \PLift_q(m, k) = \{ \ev_{\PP^m}(f) \mid f \in \FF_q[\mathbf{X}]^H_v, \forall L \in \Emb_\PP(m), \ev_{\PP^1}(f \circ L) \in \PRS_q(k) \}\,.
  \]
  Such a code will shortly be called a \emph{projective lifted code}, and its length is $\theta_{m,q} = |\PP^m| = (q^{m+1}-1)/(q-1)$.
\end{definition}

% -----------------------------------------------------------------------------
\subsection{Monomiality of projective lifted codes}
\label{subsec:monomiality}

Similarly to the affine setting, a main issue remains to give a basis of $\PLift_q(m, k)$. In this subsection, we show that projective lifted codes are monomial, and then we compute their degree set. For this purpose, we first prove Theorem~\ref{thm:glm-invariant-implies-monomial} which can be seen as a projective analogue of the \emph{monomial extraction lemma} of Kaufman and Sudan~\cite{KaufmanS08}.

\begin{theorem}
  \label{thm:glm-invariant-implies-monomial}
  Let $\calC = \ev_{\PP^m}(\mathcal{F})$ be a projective evaluation code, where $\mathcal{F}$ is a subspace of $\FF_q[\mathbf{X}]^H_v$ for some $m, v \ge 1$. Assume that $\Proj(\FF_q, m) \subseteq \Aut(\calC)$. Then $\calC$ is monomial.
\end{theorem}

Before diving straight into the proof, we first observe that lie in $\Proj(\FF_q, m)$ elements $(w_M^v, \sigma_{M^{-1}})$ where $M$ is:
\begin{itemize}
\item a diagonal isomorphism $\Diag_{\mathbf{a}}$ for any $\mathbf{a} \in (\FF_q^\times)^{m+1}$, where
\[
\begin{array}{rclc}
  \Diag_{\mathbf{a}} : & \PP^m & \to & \PP^m\\
         & (P_0 : \hdots : P_m) & \mapsto & (a_0 P_0 : \hdots : a_m P_m)
\end{array}
\]
\item an elementary switch of coordinates $s_{i,j}$ for any $0 \le i, j \le m$, $i \ne j$, where
\[
\begin{array}{rclc}
  s_{i,j} : & \PP^m & \to & \PP^m\\
         & (P_0 : \hdots : P_i : \hdots : P_j : \hdots : P_m) & \mapsto & (P_0 : \hdots : P_j : \hdots : P_i : \hdots : P_m)
\end{array}
\]
\item an elementary transvection $t_{i,j, \beta}$ for any $0 \le i, j \le m$, $i \ne j$, and $\beta \in \FF_q$, where
\[
\begin{array}{rclc}
  t_{i,j, \beta} : & \PP^m & \to & \PP^m\\
         & (P_0 : \hdots : P_i : \hdots : P_m) & \mapsto & (P_0 : \hdots : P_i + \beta P_j : \hdots : P_m)
\end{array}
\]
\end{itemize}

\begin{proof}[Proof of Theorem~\ref{thm:glm-invariant-implies-monomial}]
  Let $c = \ev_{\PP^m}(f) \in \calC$, where $f = \sum_{\mathbf{d}} f_{\mathbf{d}} \mathbf{X^d}$, and denote by $D = \Deg(f) = \{ \mathbf{d}, f_{\mathbf{d}} \ne 0 \}$. Our goal is to prove that every $\mathbf{j} \in D$ satisfies $\ev_{\PP^m}(\mathbf{X^j}) \in \calC$. The proof will consist in three main parts:
  \begin{enumerate}[label=(\roman*)]
  \item we prove that $\ev_{\PP^m}(Q_{\mathbf{j}}(\mathbf{X})) \in \calC$ for some polynomial $Q_{\mathbf{j}}({\bf X})$ such that ${\bf j} \in \Deg(Q_{\mathbf{j}})$ and $\Deg(Q_{\mathbf{j}})$ is much smaller than $\Deg(\calC)$;
  \item we analyse and rephrase $\Deg(Q_{\mathbf{j}})$, allowing us to write $Q_{\bf j}$ as $X_1^{j_1} \dots X_a^{j_a} R(X_0, X_{a+1}, \dots, X_m)$ for some multivariate polynomial $R$;
  \item we prove that, if there exists an $(m-a+1)$-variate polynomial $R$ satisfying some prescribed properties and such that $\ev_{\PP^m}(X_1^{j_1}\dots X_a^{j_a}R(X_0, X_{a+1}, \dots, X_m)) \in \calC$, then we can compute an $(m-a)$-variate polynomial $R'$ satisfying the same prescribed properties, and such that the vector
    $\ev_{\PP^m}(X_1^{j_1}\dots X_{a+1}^{j_{a+1}}R'(X_0, X_{a+2}, \dots, X_m)) \in \calC$.
  \end{enumerate}
  Reasoning inductively on the last part will conclude the proof.

  \emph{Proof of part (i).} Let $\mathbf{j} \in D$, and define
  \[
  Q_{\mathbf{j}}(\mathbf{X}) \mydef (-1)^{m+1} \sum_{\mathbf{a} \in (\FF_q^\times)^{m+1}} \Big(\prod_{i=0}^m a_i^{-j_i}\Big) (f \circ \mathrm{Diag}_{\mathbf{a}})(\mathbf{X})\,.
  \]
  Since $\mathcal{C}$ is linear and $\mathrm{Diag}_{\mathbf{a}} \in \Aut(\calC)$ for every $\mathbf{a} \in (\FF_q^\times)^{m+1}$, we see that $\ev_{\PP^m}(Q_{\mathbf{j}}(\mathbf{X})) \in \calC$. We also have:
  \[
  \begin{aligned}
    Q_{\mathbf{j}}(\mathbf{X}) &= (-1)^{m+1} \sum_{\mathbf{a} \in (\FF_q^\times)^{m+1}} \Big(\prod_{i=0}^m a_i^{-j_i}\Big) \sum_{\mathbf{d}} f_{\mathbf{d}} \,a_0^{d_0} \dots a_m^{d_m} \,\mathbf{X^d}\\
    &= (-1)^{m+1} \sum_{\mathbf{d}} f_{\mathbf{d}} \sum_{\mathbf{a} \in (\FF_q^{\times})^{m+1}} \Big(\prod_{i=0}^m a_i^{d_i-j_i}\Big) \mathbf{X}^{\mathbf{d}}\\
   &= (-1)^{m+1} \sum_{\mathbf{d}} \, f_{\mathbf{d}} \prod_{i=0}^m \Big( \underbrace{\sum_{a_i \in \FF_q^{\times}} a_i^{d_i-j_i}}_{= 0 \text{ if } d_i \not\equiv j_i \!\!\mod (q-1), \,\,\, -1 \text{ otherwise }}  \Big) \mathbf{X}^{\mathbf{d}}\\
   &= \sum_{\mathbf{d} \in E_{\mathbf{j}}} f_{\mathbf{d}}\mathbf{X^d}\,,
  \end{aligned}
  \]
  where $E_{\mathbf{j}} = \{ \mathbf{d} \in D, \mathbf{d} \equiv \mathbf{j} \mod (q-1) \}$.

  \emph{Proof of part (ii).} The code $\mathcal{\calC}$ is invariant under the action of elementary switches of coordinates. Therefore one can assume w.l.o.g. that, if exists, the $j_i$'s satisfying $j_i \in (q-1)\NN$ lie at the end of the tuple $\mathbf{j}$. Besides, by $P$-reduction and by definition of $E_{\bf j}$, we can assume that $j_i \notin \{0, q-1\}$ implies that $j_i = d_i$, except maybe for the leftmost non-zero coordinate of ${\bf d}$ and ${\bf j}$. Therefore, w.l.o.g. there exists $a \in [1,m]$ such that every $\mathbf{d} \in E_{\mathbf{j}}$ satisfies the following three properties
\begin{equation*}
  \label{eq:constraints}
  \left\{
  \begin{array}{l}
    \forall 1 \le i \le a,\text{ we have } d_i = j_i < q-1\\
    \forall a < i \le m,\text{ we have } d_i \in \{0, q-1\} \text{ and } j_i \in \{0, q-1\}\\
    d_0 = v - \sum_{i=1}^m d_i\,.
  \end{array}
  \right.
\end{equation*}

Therefore, $Q_{\mathbf{j}}(\mathbf{X})$ can be written as $X_1^{j_1}\dots X_a^{j_a} R(X_0,X_{a+1},\dots,X_m)$, where $R$ is an homogeneous polynomial of degree $v - \sum_{i=1}^a j_i$, whose monomials have partial degree either $0$ or $q-1$, for every coordinate $X_i$, $i >a$.

\emph{Proof of part (iii).} Recall that we aim to prove that $\ev_{\PP^m} (X_0^{j_0}X_1^{j_1}\dots X_m^{j_m}) \in \calC$, and we know that $\ev_{\PP^m}(Q_{\mathbf{j}}(\mathbf{X})) \in \calC$.  Our strategy is to proceed inductively, from $i = a$ to $m$, by proving there exists an $(m-i+1)$-variate polynomial $R_i$  such that ${\bf j} \in \Deg(X_1^{j_1}\dots X_i^{j_i} R_i(X_0,X_{i+1},\dots,X_m))$ and $\ev_{\PP^m}(X_1^{j_1}\dots X_{i}^{j_i}R_i(X_0, X_{i+1}, \dots, X_m)) \in \calC$. Notice that step $i=a$ has been proved in part (ii), and that step $i=m$ concudes the proof. Hence there remains to prove the induction step.

Write $R_i = R_i' + X_{i+1}^{q-1}R_i''$, where polynomials $R'_i$ and $R''_i$ do not depend on $X_{i+1}$. Also denote by $S_i = X_1^{j_1}\dots X_{i}^{j_i}R_i(X_0, X_{i+1}, \dots, X_m)$,  and assume that $\ev_{\PP^m}(S_i) \in \calC$ and ${\bf j} \in \Deg(S_i)$. If $R''_i = 0$, then the induction step is proved. Otherwise:

\emph{$\bullet$ 1st case: $j_{i+1} = 0$.} Since $\sum_{\beta \in \FF_q} (X_{i+1} + \beta X_0)^{q-1} = -X_0^{q-1}$ (see Lemma~\ref{lem:technical-sum} in the appendix), we get
\[
\begin{aligned}
\sum_{\beta \in \FF_q} S_i&(X_0,\dots,X_{i+1} + \beta X_0, \dots, X_m)\\
&= X_1^{j_1}\dots X_i^{j_i} \sum_{\beta \in \FF_q} \left(R'_i(X_0, X_{i+2}, \dots, X_m) + (X_{i+1} + \beta X_0)^{q-1} R''_i(X_0, X_{i+2}, \dots, X_m)\right)\\
&= \Big( \sum_{\beta \in \FF_q} X_1^{j_1}\dots X_i^{j_i} R'_i(X_0, X_{i+2}, \dots, X_m) \Big) - X_1^{j_1}\dots X_i^{j_i} X_0^{q-1} R''_i(X_0,X_{i+2},\dots,X_m)\\
&= - X_1^{j_1}\dots X_i^{j_i} X_0^{q-1} R''_i(X_0,X_{i+2},\dots,X_m) \,.
\end{aligned}
\]
By linearity and stability of $\calC$ under elementary transvections, $\ev_{\PP^m}(S_i) \in \calC$ ensures that the word $\ev_{\PP^m}(X_1^{j_1}\dots X_i^{j_i} X_0^{q-1} R''_i(X_0,X_{i+2},\dots,X_m)) \in \calC$. We conclude by defining $R_{i+1} = X_0^{q-1}R''_i$.

\emph{$\bullet$ 2nd case: $j_{i+1} = q-1$.}
Since $\sum_{\beta \in \FF_q} (\beta X_{i+1} + X_0)^{q-1} = -X_{i+1}^{q-1}$, we get
\[
\begin{aligned}
\sum_{\beta \in \FF_q} S_i&(X_0,\dots,X_{i+1} + \beta X_0, \dots, X_m)\\
&= X_1^{j_1}\dots X_i^{j_i} \sum_{\beta \in \FF_q} \left(R'_i(X_0,X_{i+2},\dots,X_m) + (\beta X_{i+1} +  X_0)^{q-1} R''_i(X_0,X_{i+2},\dots,X_m)\right)\\
&= - X_1^{j_1}\dots X_i^{j_i} X_{i+1}^{q-1} R''_i(X_0,X_{i+2},\dots,X_m)\,.
\end{aligned}
\]
Similarly to the first case, we can conclude by defining $R_{i+1} = R''_i$.
\end{proof}

Projective lifted codes can be proved invariant under $\Proj(\FF_q, m)$.
\begin{lemma}\label{lem:glm-stab-plift}
  Let $k \le q-1$, $m \ge 1$ and $\calC = \PLift_q(m, k)$. Then $\Proj(\FF_q, m) \subseteq \mathrm{Aut}(\calC)$. Said differently,
  \[
  \forall c = \ev_{\PP^m}(f) \in \calC,\, \forall M \in \mathrm{GL}_{m+1}(\FF_q),\, \ev_{\PP^m}(f \circ M) \in \calC\,.
  \]
\end{lemma}
\begin{proof}
  It is sufficient to notice that, for every $L \in \Emb_\PP(m)$ and every $M \in \mathrm{GL}_{m+1}(\FF_q)$, the map $M \circ L$ also lies in $\Emb_\PP(m)$.
\end{proof}

As a corollary,
\begin{corollary}
  Every projective lifted code is monomial.
\end{corollary}

% -----------------------------------------------------------------------------
\subsection{Degree sets of lifted codes}
\label{subsec:degree-sets}

A natural question is now to determine the degree set of $\PLift_q(m, k)$. Let us first recall that affine  lifted codes have the following degree sets (see equation~\eqref{eq:degree-set-lift}):
\[
\mathrm{ADeg}_q(m, k) \mydef \Deg(\Lift_q(m,k)) = \{ \mathbf{d} \in B^m_\infty(k), \forall \mathbf{e} \le_p \mathbf{d}, \underline{|\mathbf{e}|} \le k \}\,.
\]
Similarly, we define $ \mathrm{PDeg}_q(m, k) \mydef \Deg(\PLift_q(m,k))$.

In this subsection, we state a few links between degree sets of affine and projective lifted codes. Propositions~\ref{prop:d0!=0} and~\ref{prop:d0=0} show that $\mathbf{d} \in \mathrm{PDeg}(m,k)$ can be sent either to $\mathrm{ADeg}(m,k-1)$ or to $\mathrm{PDeg}(m-1,k)$, according to the value of $d_0$. Then, in Theorem~\ref{theo:bij-degree-sets} we derive a recursive formula on the degree sets of affine/projective lifted codes, which translates into another recursive formula on the dimension of these codes (Corollary~\ref{coro:recursive-dimension}).

\begin{proposition}
  \label{prop:d0!=0}  
  Let $v = k + (m-1)(q-1)$ for $1 \le k \le q-1$ and $m \ge 2$. Let also $\mathbf{d} = (d_0, \mathbf{d^*}) \in S^{m+1}(v)$ such that $d_0 \ne 0$. Then:
  \[
  \ev_{\PP^m}(X_0^{d_0} \mathbf{X}^{\mathbf{d^*}}) \in \PLift_q(m, k) \iff \ev_{\AA^m}(\mathbf{X^{d^*}}) \in \Lift_q(m, k-1)\,,
  \]
  or, equivalently,  
  \[
  \mathbf{d} = (d_0, \mathbf{d^*}) \in \mathrm{PDeg}_q(m, k) \iff \mathbf{d^*} \in \mathrm{ADeg}_q(m, k-1)\,.
  \]
\end{proposition}
\begin{proof}  
  ($\Rightarrow$). Let $\ev_{\PP^m}(X_0^{d_0}{\bf X^{d^*}}) \in \PLift_q(m, k)$ with $d_0 \ne 0$. Let also $L^* \in \Emb_\AA(m)$; we need to prove that $\ev_{\AA^1}({\bf X^{d^*}} \circ L^*) \in \RS_q(k-1)$. Let us define $L =(L_0, \dots, L_m) \in \Hom(\FF_q^2, \FF_q^{m+1})$ as follows: the $m$ last coordinates $(L_1, \dots, L_m) = L^*$, and the first coordinate $L_0$ is chosen between $L_0(S, T) = S$ and $L_0(S, T) = T$, in order to have $\rank(L) = 2$. Now assume w.l.o.g. that $L_0(S,T) = S$.\footnote{Two points could be clarified here. First, if the linear map $(S, L^*(S,T))$ as rank $1$, by definition of $\Emb_{\AA}(m)$ the linear map $(T, L^*(S,T))$ has rank $2$. Second, the choice $L_0(S,T) = S$ can be done since $\PRS_q(k)$ is invariant under $\Proj(\FF_q,1)$.} Then, 
\begin{equation}
  \label{eq:d0-not-0}
\ev_{\PP^1}(X_0^{d_0}{\bf X^{d^*}} \circ L) = \ev_{\PP^1}(X_0{\bf X^{d^*}} \circ L) = \ev_{\PP^1}(S\,.\, ({\bf X^{d^*}} \circ L^*) (S, T)) \in \PRS_q(k)
\end{equation}
since $X_0^{d_0}\mathbf{X^{d^*}}$ and $X_0 \mathbf{X^{d^*}}$ evaluate identically. Besides, for any homogeneous polynomial $P(S, T)$, we know that
\begin{equation}
  \label{eq:equiv-RS-PRS}
  \ev_{\PP^1}(S\,.\,P(S,T)) \in \PRS_q(k) \iff \ev_{\AA^1}(P(1,T)) \in \RS_q(k-1)\,.
\end{equation}
Applying this to $P(S, T) = ({\bf X^{d^*}} \circ L^*)(S, T)$, we get our result.

  ($\Leftarrow$). Let $\ev_{\AA^m}({\bf X^{d^*}}) \in \Lift_q(m, k)$ and $L \in \Emb_\PP(m)$. Let also $d_0 \ne 0$ such that $(d_0,{\bf d^*}) \in S^{m+1}(v)$. We need to prove that $\ev_{\PP^1}(X^{d_0}{\bf X^{d^*}} \circ L) \in \PRS_q(k)$. If $L_0 = 0$, then the result holds since $0 \in \PRS_q(k)$. Otherwise, it is worthwhile to notice that, since $\PRS_q(k)$ is invariant under $\Proj(\FF_q, 1)$, we can assume w.l.o.g. that $L_0(S, T) = S$. Define $L^* = (L_1, \dots, L_m)$, which lies in $\Emb_{\AA}(m)$ by definition. Therefore $\ev_{\AA^1}({\bf X^{d^*}} \circ L^*) \in \RS_q(k-1)$, and using~\eqref{eq:d0-not-0} and~\eqref{eq:equiv-RS-PRS}, we get our claim.
\end{proof}

\begin{proposition}
  \label{prop:d0=0}
  Let $v = k + (m-1)(q-1)$ for $1 \le k \le q-1$ and $m \ge 2$. Let also $\mathbf{d} = (d_0, {\bf d^*}) \in S^{m+1}(v)$, ans assume that $d_0 = 0$. Then:
  \[
  \ev_{\PP^m}({\bf X^d}) \in \PLift_q(m, k) \iff \ev_{\PP^{m-1}}({\bf X^{d^*}}) \in \PLift_q(m-1, k)\,,
  \]
  or equivalently,  
  \[
  \mathbf{d} = (0, \mathbf{d^*}) \in \mathrm{PDeg}_q(m, k) \iff \overline{\mathbf{d^*}} \in \mathrm{PDeg}_q(m-1, k)\,.
  \]
\end{proposition}
\begin{proof}
  ($\Rightarrow$). Let $\ev_{\PP^m}({\bf X^d}) \in \PLift_q(m, k)$ where ${\bf d} = (d_0, {\bf d^*})$ and $d_0 = 0$. Let also $L' \in \Emb_\PP(m-1)$; we need to prove that $\ev_{\PP^1}({\bf X^{d^*}} \circ L') \in \PRS_q(k)$. Any $L_0 \in \Hom(\FF_q^2, \FF_q)$ extends $L'$ to $L = (L_0, L') \in \Emb_{\PP}(m)$. Therefore, 
\begin{equation}
  \label{eq:d0-eq-0}
\ev_{\PP^1}({\bf X^{d^*}} \circ L') = \ev_{\PP^1}(X_0^0{\bf X^{d^*}} \circ L) = \ev_{\PP^1}({\bf X^d} \circ L) 
\end{equation}
lies in $\PRS_q(k)$ since  $\ev_{\PP^m}({\bf X^d}) \in \PLift_q(m, k)$. 

($\Leftarrow$). Let ${\bf d} = (d_0, {\bf d^*}) \in S^{m+1}(v)$ with $d_0 = 0$, and assume that $\ev_{\PP^{m-1}}({\bf X^{d^*}}) \in \PLift_q(m-1, k)$. Let also $L \in \Emb_\PP(m)$; we need to prove that $\ev_{\PP^1}({\bf X^d} \circ L) \in \PRS_q(k)$.  Write $L = (L_0, L')$. If $L' \in \Emb_\PP(m-1)$, then the result follows using~\eqref{eq:d0-eq-0}. The case $\rank L' = 1$ is a bit trickiest. Since $\Proj(\FF_q, 1)$ let the code $\PRS_q(k)$ invariant, we can assume w.l.o.g. that, seen as a $2 \times m$ matrix, the second row of $L'$ is zero. In other words, $L'(S, T)$ can be written $(\ell_1S, \dots, \ell_mS)$ with some non-zero $(\ell_1, \dots, \ell_m) \in \FF_q^m$. Therefore, $({\bf X^d} \circ L)(S, T) = \alpha S^{|{\bf d^{*}}|}$ with $\alpha \in \FF_q$, and $\ev_{\PP^1}({\bf X^d} \circ L) = \alpha \ev_{\PP^1}(S^{|{\bf d^{*}}|}) \in \PRS_q(k)$ since the $P$-reduction of $(|{\bf d^{*}}|, 0)$ is $(k, 0)$ which lies in $\Deg(\PRS_q(k))$.
\end{proof}

\begin{theorem}
  \label{theo:bij-degree-sets}
  For every $m \ge 2$ and $1 \le k \le q-1$, there is a bijection between $\mathrm{PDeg}(m, k)$ and $\mathrm{PDeg}(m-1,k) \cup \mathrm{ADeg}(m, k-1)$.
\end{theorem}

\begin{proof}
  According to Propositions~\ref{prop:d0!=0} and~\ref{prop:d0=0}, this bijection is given by:
\[
\mathbf{d} = (d_0, \mathbf{d^*}) \mapsto \left\{\begin{array}{ll} \mathbf{d^*} \in \mathrm{ADeg}(m, k-1) & \text{if } d_0 \ne 0 \\ \overline{\mathbf{d^*}} \in \mathrm{PDeg}(m-1, k) & \text{otherwise.}\end{array}\right.
\]
\end{proof}

A recursive formula on the dimension of lifted codes follows.
\begin{corollary}
  \label{coro:recursive-dimension}
  Let $m \ge 2$ and $1 \le k \le q-1$. Then,
  \[
  \dim \PLift_q(m, k) = \dim \PLift_q(m-1, k) + \dim \Lift_q(m, k-1)
  \]
\end{corollary}

One can also check that $\PLift_q(1, k) = \PRS_q(k)$ and $\Lift_q(1, k-1) = \RS_q(k-1)$. Therefore we also get:
\begin{corollary}
  Let $m \ge 1$ and $1 \le k \le q-1$. Then,
  \[
  \dim \PLift_q(m, k) = \sum_{j = 1}^m \dim \Lift_q(j, k-1) + 1\,.
  \]
\end{corollary}

Finally, if one would like to explicit $\PDeg_q(m, k)$, one could use iteratively the bijective map given in Theorem~\ref{theo:bij-degree-sets} and the characterisation of $\ADeg_q(j, k-1)$, $1 \le j \le m$, given previously. For ${\bf d} = (d_0, \dots, d_m) \in S^{m+1}(v)$, define $\iota({\bf d})$ the minimum $i$ such that $d_i \ne 0$, and $\eta({\bf d}) = (d_{\iota({\bf d})+1}, \dots, d_m) \in S^{m-\iota({\bf d})}$. We then obtain
\begin{corollary}
  Let $m \ge 2$ and $1 \le k \le q-1$. Denote by $v = k + (m-1)(q-1)$. Then,
  \[
  \PDeg_q(m, k) = \{ \overline{\bf d} \in S^{m+1}(v), \forall {\bf e} \le_p \eta(\overline{\bf d}), \underline{|{\bf e}|} \le k-1 \}\,.
  \]
\end{corollary}

One also can see that any ${\bf d} \in S^{m+1}(k)$ can be lifted in $\Deg(\PLift_q(m, k))$ by adding $(q-1)(m-1)$ to the leftmost non-zero coordinate of ${\bf d}$. Hence a corollary is the projective analogue of equation~\eqref{eq:lift-inclusion}.
\begin{corollary}
  Let $1 \le k \le q-1$ and $v = k + (m-1)(q-1)$. Then we have:
  \[
  \mathrm{PRM}_q(m, k) \subseteq \PLift_q(m, k) \subseteq \mathrm{PRM}_q(m, v)\,,
  \]
  where the inclusion are taken up to diagonal isomorphisms of codes.
\end{corollary}

\begin{example}
  We give here the smallest example of projective lifted code that is not isomorphic to any projective Reed-Muller code. Let $q=4$, $m=2$ and $k=3$, giving $v = k + (m-1)(q-1) = 6$. The corresponding projective Reed-Muller code has length $q^2 + q + 1 = 21$, dimension $\binom{m+d+1}{d} = 10$, and admits
\[
\begin{aligned}
  D =\,\, & \{ (3,0,0), (2,1,0), (2,0,1), (1,2,0), (1,1,1), (1,0,2), (0,3,0),\\
  &\quad (0,2,1), (0,1,2), (0,0,3) \}
\end{aligned}
\]
as a degree set. A computation shows that $\PLift_q(m, k)$ is given by the following degree set:
\[
\begin{aligned}
  D_L =\,\, & \{ (6, 0, 0), (5, 1, 0), (5, 0, 1), (4, 2, 0), (4, 1, 1), (4, 0, 2), (0, 6, 0), \\
  &\quad (0, 5, 1), (0, 4, 2), (0, 0, 6), (2, 2, 2)\}\,.
\end{aligned}
\]
One observes that $D_L = D' \cup \{ (2,2,2) \}$, where $D'$ is obtained by adding $q-1 = 3$ to the leftmost non-zero coordinate of every ${\bf d} \in D$. Besides, the affine lifted code $\Lift_q(m, k-1)$ has the following degree set:
\[
D_A = \{ (0, 0), (1, 0), (0 ,1), (2, 0), (1, 1), (0, 2), (2, 2) \}\,.
\]
We see that $D_A$ corresponds to elements ${\bf d} \in D_L$ such that $d_0 \ne 0$, then punctured on their first coordinate. We also remark that the remaining elements $\{(0, 6, 0), (0, 5, 1), (0, 4, 2), (0, 0, 6) \}$, being at first punctured on their first coordinate and then $P$-reduced, give the degree set $\{ (3, 0), (2, 1), (1, 2), (0,3)\}$ of $\PRS_q(k)$.

Finally, notice that the extra degree $(2,2,2)$ which makes $\PLift_q(m, k)$ non-isomorphic to $\PRM_q(m, k)$ corresponds to the codeword $c = \ev_{\PP^2}(X^2Y^2Z^2)$. A tedious computation can then confirm that any embedding $(S,T) \mapsto (a_0S + b_0T, a_1S + b_1T, a_2S + b_2T)$ sends $c$ to a projective Reed-Solomon codeword.
\end{example}

% =============================================================================
\section{Local correction}
\label{sec:local-correction}

% -----------------------------------------------------------------------------
\subsection{Definitions}

This section is devoted to local correcting properties of projective lifted codes. After Guo \emph{et al.}'s work~\cite{GuoKS13}, we know that affine lifted codes are (perfectly smooth) locally correctable codes. We first recall this notion.

\begin{definition}[locally correctable code]
  Let $\Sigma$ be a finite set, $2 \le \ell \le k \le n$ be integers, and $\delta, \epsilon \in [0,1]$. A code $\calC : \Sigma^k \to \FF_q^n$ is $(\ell, \delta, \epsilon)$--locally correctable if and only if there exists a randomized algorithm $\calD$ such that, for every input $i \in [1,n]$ we have:
  \begin{itemize}
  \item for all $c \in \calC$ and all $y \in \FF_q^n$, if $|\{ j \in [1,n], y_j \ne c_j \}| \le \delta n$, then
    \[
    \mathbb{P}( \calD^{(y)}(i) = c_i ) \ge 1 - \epsilon\,,
    \]
    where the probability is taken over the internal randomness of $\calD$;
  \item $\calD$ reads at most $\ell$ symbols $y_{q_1}, \dots, y_{q_\ell}$ of $y$.
  \end{itemize}
  Notation $\calD^{(y)}$ refers to the fact that $\calD$ has oracle access to single symbols $y_{q_j}$ of the word $y \in \FF_q^n$. The parameter $\ell$ is called the \emph{locality} of the code. Moreover, the code $\calC$ is said \emph{perfectly smooth} if on arbitrary input $i$, each individual query of the probabilistic algorithm $\calD$ is uniformly distributed over the coordinates of the word $y$.
\end{definition}

By definition of projective lifted codes, if $c = \ev_{\PP^m}(f) \in \PLift_q(m, k)$, then $\ev_{\PP^1}(f \circ L) \in \PRS_q(k)$ for all $L \in \Emb_{\PP}(m)$. In Remark~\ref{rem:subword} we noticed that $\ev_{\PP^1}(f \circ L)$ is not a subword of $\ev_{\PP^m}(f)$. Nevertheless, denoting $v = k + (m-1)(q-1)$, there still exists $w^v_L \in (\FF_q^{\times})^{q+1}$ such that $w^v_L \star \ev_{\PP^1}(f \circ L)$ is such a subword. Moreover, given $L$ and the standard representation of points in $\PP^1$, each coordinate of $w_L^v$ is expressed as a $v$-th power of a linear combination of $calO(m)$ $\FF_q$-symbols  (hence, it is also a $k$-th power since every $x \in \FF_q$ satisfies $x^q = x$). Therefore, $w^v_L \in \FF_q^{q+1}$ can be computed in $\calO(mq \log k)$ operations over $\FF_q$.

To sum up we get:
\begin{lemma}
  Let $c = \ev_{\PP^m}(f) \in \PLift_q(m, k)$ and $L \in \Emb_\PP(m)$. There exists a deterministic algorithm which computes $\ev_{\PP^1}(f \circ L)$ from $c$ and $L$, with $q+1$ queries to $c$ and $\calO(mq \log k)$ operations in $\FF_q$.
\end{lemma}

% -----------------------------------------------------------------------------
\subsection{Local correcting algorithms}
\label{subsec:local-correcting-algorithms}

For convenience, we fix a projective lifted code $\PLift_q(m, k)$, and we denote by $n = \theta_{m,q}$ its length. We also denote by $\infty$ the point $(0 : 1) \in \PP^1$, and for a given $\mathbf{P} \in \PP^m$,
\[
\Emb_\PP(m, \mathbf{P}) \mydef \{ L \in \Emb_\PP(m), \ev_{\infty}(L) = \mathbf{P} \}\,
\]
is the set of embeddings having $\mathbf{P}$ as image of the point at infinity. We denote by $L(\PP^1)$ the set $\{ \ev_{\mathbf{Q}}(L), \mathbf{Q} \in \PP^1 \}$.

In this subsection, we present a generic local correcting algorithm for projective lifted codes. This algorithm depends on a parameter $s \in [k+1, q]$, and it informally works as follows: (i) pick at random $s$ points on a random projective line of $\PP^m$, (ii) correct the associated noisy $\PRS_q(k)$ codeword, and (iii) output the desired corrected symbol. This is somewhat a projective analogue of a generalization of the two well-known Reed-Muller local correcting algorithms (see Yekhanin's survey~\cite[Section 2.2]{Yekhanin12}), since we do not restrict $s \in \{k+1, q\}$.

Let us assume we have access to a query generator $\calR_s$, for $k+1 \le s \le q$, with parameters $\mathbf{P} \in \PP^m$ and $L \in \Emb_\PP(m, \mathbf{P})$, such that for all $\mathbf{P}$ we have:
\begin{itemize}
\item for all $L \in \Emb_\PP(m, \mathbf{P})$, if $S \leftarrow \calR_s(\mathbf{P}, L)$, then $S \subseteq L(\PP^1)$ and $|S| = s$,
\item $\forall \mathbf{Q} \in \PP^m, \mathrm{Pr}_{\substack{L \leftarrow \Emb_\PP(m, \mathbf{P})\\S \leftarrow \calR_s(\mathbf{P}, L)}} [ \mathbf{Q} \in S ] = s/n$.
\end{itemize}
Such a query generator can be implemented, as we show in Appendix~\ref{app:query-generator}.

We also assume to have at our disposal an error-and-erasure correcting algorithm for $\PRS_q(k)$, which corrects $q+1-s$ erasures and up to $t = \lfloor \frac{s-k-1}{2} \rfloor$ errors (we recall that $\PRS_q(k)$ is an MDS code and $\dim \PRS_q(k) = k+1$). Call $\Corr^{\PRS}_s$ this correcting algorithm, and see~\cite{ReedTM79} for a simple example.

\begin{algorithm}[h!]
  \KwData{
    a noisy word $y \in \FF_q^{\PP^m}$ such that $d(c, y) \le \delta n$ for some codeword $c \in \calC$, and
    a point $\mathbf{P} \in \PP^m$ where to correct $y_{\bf P}$
    }
  \KwResult{the symbol $c_{\mathbf{P}}$}
  Pick uniformly at random $L \leftarrow \Emb_\PP(m, \mathbf{P})$, and pick  $S = \{ \mathbf{Q}_1, \dots, \mathbf{Q}_s \} \leftarrow \calR_s(\mathbf{P}, L)$.\\
  Denote by $\mathbf{R}_j = \ev_{\mathbf{Q}_j}(L)$. Collect $y_{\mathbf{R}_1}, \dots, y_{\mathbf{R}_s}$ and define $y' \in \FF_q^{\PP^1}$ by:
  \[
  y'_{\mathbf{Q}} \mydef \left\{\begin{array}{ll} w^v_{L,\mathbf{Q}_j} \cdot y_{\mathbf{R}_j} & \text{if } \mathbf{Q} = \mathbf{Q}_j \text{ for some } j\\ \perp & \text{otherwise.} \end{array}\right. 
  \]
  where $w^v_{L,\mathbf{Q}_j}$ is the symbol indexed by $\mathbf{Q}_j$ in the tuple $w^v_L$ given in Remark~\ref{rem:subword}\\
  Use $\Corr_s^{\PRS}$ to correct $y'$ as a noisy codeword of $\PRS_q(k)$.\\
  If this succeeds, return $y'(\infty)$. Otherwise return $\perp$.
  \caption{\label{algo:local-decoding}A generic local correcting algorithm $\Corr^{\PLift}_s$ for $\calC \mydef \PLift_q(m, k)$}
\end{algorithm}

\begin{theorem}
  Let $k+1 \le s \le q$ and $t = \lfloor \frac{s - k - 1}{2} \rfloor$. For every $\delta \le \frac{t+1}{2s}$, the code $\PLift_q(m, k)$ is a perfectly smooth $(s, \delta, \frac{\delta s}{t+1})$-locally correctable code using Algorithm~\ref{algo:local-decoding}.
\end{theorem}
\begin{proof}
  Let us analyse Algorithm~\ref{algo:local-decoding}. Concerning the locality, the algorithm indeed makes $\ell = s$ queries to $y$. Besides, it is smooth due to our assumption on $\calR_s$. Now let us focus on the correctness.

  Let $y \in \FF_q^{\PP^m}$ such that $d(y, c) \le \delta n$ for some $c \in \PLift_q(m, k)$. Denoting $E = \{ \mathbf{Q} \in \PP^m, c_{\mathbf{Q}} \ne y_{\mathbf{Q}} \}$, we have $|E| \le \delta n$. By definition of the correcting algorithm of $\PRS_q(k)$, the output value is correct as long as $|E \cap S| \le t$. 
  Let us bound this probability. Using Markov's inequality,
  \[
   \mathrm{Pr}_{L, S} [ |E \cap S| \le t ] = 1 - \mathrm{Pr}_{L, S} [ |E \cap S| \ge t + 1 ] \ge 1 - \EE_{L, S} [ |E \cap S| ]/(t+1) \,.
  \]
  By linearity, we get:
  \[
  \EE_{L, S} [ |E \cap S| ] = \sum_{e \in E} \mathrm{Pr}_{L, S}[ e \in S ] = \sum_{e \in E}\frac{s}{n} = \delta n \frac{s}{n} = \delta s\,.
  \]
  Hence, 
  \[
  \mathrm{Pr}_{L, S} [ \text{ algorithm succeeds } ] \ge 1 - \frac{\delta s}{t+1}\,.
  \]
\end{proof}

We exhibit the two extreme instances which correspond to the well-known correcting algorithms of Reed-Muller codes presented in~\cite{Yekhanin12} for instance. The first one picks the least possible number of symbols, but assumes few errors on the corrupted codeword.

\begin{corollary}[$s=k+1$]
  For every $\delta \le \frac{1}{2(k+1)}$, the code $\PLift_q(m, k)$ is a perfectly smooth $(k+1, \delta, \delta(k+1))$-locally correctable code.
\end{corollary}
\begin{proof}
$s=k+1$ implies $t=0$.
\end{proof}

The second one achieves local correction under a constant fraction of errors on the corrupted codeword. 

\begin{corollary}[$s=q$]
  Let $\tau = \frac{k+1}{q}$. For every $\delta \le \frac{1}{4}(1 - \tau)$, the code $\PLift_q(m, k)$ is a perfectly smooth $(q, \delta, \frac{2\delta}{1 - \tau})$-locally correctable code.
\end{corollary}
\begin{proof}
$s=q$ implies $t = \lfloor \frac{q(1 - \tau)}{2} \rfloor$, hence $t+1 \ge \frac{q(1 - \tau)}{2}$.
\end{proof}

\begin{remark}
  In Algorithm~\ref{algo:local-decoding}, we can avoid to compute the tuple $w^v_L$. Indeed, it can be proved that for every projective line $S \subset \PP^m$ and every point $\mathbf{P} \in S$, there exists an $L \in \Emb_\PP(m)$ such that $L(\PP^1) = S$, $\mathbf{P} \in \{ \ev_{(1:0)}(L), \ev_{\infty}(L) \}$ and $w^v_L = (1, \dots, 1)$.
\end{remark}

\begin{remark}
  Local testability of affine lifted codes was also proved by Guo \emph{et al.}~\cite{GuoKS13}. Once again, their results rely on the work of Kaufman and Sudan~\cite{KaufmanS08} regarding the testability of some families of affine-invariant codes. Though, projective lifted codes cannot be proved locally testable the same manner, since their automorphism group is slightly different. Though, this issue is worth addressing in a future work.
\end{remark}

% =============================================================================
\section{Puncturing and shortening relations between affine and projective lifted codes}
\label{sec:relations}

In this section we aim at showing links between affine and projective lifted codes through shortening and puncturing operations on codes.

% -----------------------------------------------------------------------------
\subsection{Motivation and similar results}

The embedding of both $\PP^{m-1}$ and $\AA^m$ into $\PP^m$ issues the relation between affine and projective Reed-Muller codes. Indeed, the hyperplane at infinity $\Pi_\infty \mydef \{ \mathbf{P} \in \PP^m, P_0 = 0 \}$ defines a restriction map
\[
\begin{array}{rclc}
  \pi : &\FF_q^{\PP^m} &\to      &\FF_q^{\Pi_{\infty}}\\
  ~     &c           &\mapsto  &c_{|\Pi_{\infty}}\,.
\end{array}
\]
Map $\pi$ induces a surjective map $\PRM_q(m,k) \twoheadrightarrow \PRM_q(m-1,k)$ by seeing $\Pi_\infty$ as the projective space $\PP^{m-1}$. Indeed, every $m$-variate homogeneous polynomial of degree $k$ can be also considered as an $(m+1)$-variate homogeneous polynomial of same degree (in which the new variable, denoted $X_0$, does not appear).

Besides, the vector space $K \mydef \ker\big(\PRM_q(m,k) \twoheadrightarrow \PRM_q(m-1,k)\big)$ consists in evaluation vectors of homogeneous polynomials $P \in \FF_q[X_0,\dots,X_m]^H_k$ such that $X_0$ divides $P$. That is,
\[
K = \{ \ev_{\PP^m}(X_0Q), Q \in \FF_q[X_0,\dots,X_m]^H_{k-1} \}\,.
\]

Now, restricting $K$ to coordinates in $(\PP^m \setminus \Pi_{\infty}) \simeq \AA^m$ leads to a vector space isomorphic to $\RM_q(m,k-1)$, since $X_0$ evaluates to $1$ on every affine point of $\PP^m$.

To sum up, we have the following short exact sequence:
\[
0 \to \RM_q(m,k-1) \to \PRM_q(m,k) \xrightarrow{\pi} \PRM_q(m-1,k) \to 0\,.
\]

From a coding theory point of view, it may be more comfortable by viewing this sequence in the terminology of puncturing and shortening. Indeed, up to isomorphism, the surjective map $\pi$ corresponds to the puncturing of $\PRM_q(m,k)$ on coordinates lying in $\AA^m \subset \PP^m$, while the injection $\RM_q(m,k-1) \xhookrightarrow{} \PRM_q(m,k)$ corresponds to its shortening on $\PP^{m-1} \subset \PP^m$.

A very similar exact sequence holds for the codes coming from the \emph{block designs} of incidences between points and hyperplanes. Let us denote by $\calC(D)$ the code whose dual code is generated by the incidence matrix of a block design $D$ (we refer to~\cite{Stinson04, AssmusK92} for details on block designs and their associated codes). Let also $\mathrm{AG}_t(m,q)$ and $\mathrm{PG}_t(m,q)$ be respectively the designs of points and $t$-flats in affine and projective spaces of dimension $m > t$ over $\FF_q$. Then it holds that
\[
0 \to \calC(\mathrm{AG}_1(m, q)) \to \calC(\mathrm{PG}_1(m, q)) \xrightarrow{\pi} \calC(\mathrm{PG}_1(m-1, q)) \to 0\,.
\]
This result is presented by Assmus and Key in~\cite[Theorem 5.7.2]{AssmusK92} for the dual of these codes, but it remains true for the codes we consider, since duality of codes preserves such short sequences.

In this section, our goal is to prove similar results for lifted codes.

% -----------------------------------------------------------------------------
\subsection{Shortening and puncturing projective lifted codes}

We recall that $\Pi_\infty$ denotes the hyperplane of $\PP^m$ defined by $X_0 = 0$.

\begin{theorem}
  \label{theo:puncturing-shortening}
  Let $m \ge 1$, $1 \le k \le q-1$, and $v = k + (m-1)(q-1)$. Let also
  \[
  \mathcal{S} \mydef \{ c_{|\AA^m} \mid c \in \PLift_q(m, k) \text{ and } c_{\mathbf{P}} = 0, \forall \mathbf{P} \in \Pi_\infty \}
  \]
  be the shortening of $\PLift_q(m, k)$ at the coordinates indexed by points in $\Pi_\infty$, and 
  \[
  \mathcal{P} \mydef \{ c_{|\Pi_\infty} \mid c \in \PLift_q(m, k) \}
  \]
  be the puncturing of $\PLift_q(m, k)$ at the coordinates indexed by points in $\PP^m \setminus \Pi_\infty$. Then
  \[
  \mathcal{S} = \Lift_q(m, k-1) \quad \text{ and } \quad \mathcal{P} = \PLift_q(m-1, k)\,.
    \]
\end{theorem}
\begin{proof}  \emph{(i) Proof of $\mathcal{S} = \Lift_q(m, k-1)$.} Let  $c = \ev_{\AA^m}(\mathbf{X^d}) \in \Lift_q(m, k-1)$ and extend it to $c' = \ev_{\PP^m}(X^{d_0}_0\mathbf{X^d}) \in \PLift_q(m, k)$, with $d_0 = v - |\mathbf{d}| > 0$. We notice that $c'$ vanishes on the coordinates corresponding to points in $\Pi_\infty$, and that $c' = c$ elsewhere, hence $c \in \mathcal{S}$.

Conversely, let $c \in \mathcal{S}$. There exists $f \in \FF_q[\mathbf{X}]^H_v$ such that $c' = \ev_{\PP^m}(f)$ satisfies $c' = 0$ over all coordinates of $\Pi_\infty$, and $c' = c$ elsewhere. It means that the polynomial $f$ vanishes on the whole projective hyperplane $\Pi_\infty$ given by $X_0 = 0$. Therefore $f$ vanishes over the hyperplane $\Pi'_\infty$ of the affine space $\AA^{m+1}$ given by $X_0 = 0$.

The previous remark makes sense since we can apply the Combinatorial \emph{Nullstellensatz} proved by Alon in~\cite{Alon99} (see Theorem~\ref{thm:combinatorial-nullstellensatz} in the appendix). This result asserts that, if $W = \prod_{i=0}^m W_i \subseteq \FF_q^{m+1}$ and $\deg f = \sum^m_{i=0} t_i$ with each $t_i < |W_i|$, then $f(W) = \{ 0 \}$ implies $f_{\bf t} = 0$, where $f_{\bf t}$ denotes the coefficient of the monomial $\mathbf{X^t}$ in $f$. In our context, let $W = \{0\} \times \FF_q^m$ and $\mathbf{t}$ satisfy $t_0 = 0$ and $t_i \le q-1$ for all $i>0$. The Combinatorial \emph{Nullstellensatz} then shows that  $\mathrm{Coeff}(f, \mathbf{X^t}) = 0$. Therefore every monomial in $f$ must be divisible by $X_0$. Said differently, $f$ is a sum of monomials $\mathbf{X^d}$ with ${\bf d}$ such that $d_0 \ne 0$, and Proposition~\ref{prop:d0!=0} then shows that $\ev_{\AA^m}(f) \in  \Lift_q(m, k-1)$.

\emph{(ii) Proof of $\mathcal{P} = \PLift_q(m-1, k)$.} First, $\PLift_q(m-1,k) \subseteq \mathcal{P}$, since $c = \ev_{\PP^{m-1}}(\mathbf{X^d}) \in {\PLift_q(m-1, k)}$ can be extended to $c' = \ev_{\PP^m}(\mathbf{X^{d'}}) \in \PLift_q(m, k)$, where we define $\mathbf{d'}$ by adding $q-1$ to the leftmost non-zero coordinate of $\mathbf{d}$.

Conversely, let $c' = \ev_{\PP^m}(f) \in \PLift_q(m, k)$ such that $c'_{|\Pi_\infty} \in \calP \setminus \{0\}$. Let $\mathbf{X^d}$ be a monomial in $f$. If $d_0 \ne 0$, then $\ev_{\PP^m}(\mathbf{X^d})_{|\Pi_\infty} = 0$, hence one can assume that every monomial ${\bf X^d}$ composing $f$ satisfies $d_0 = 0$. Using Proposition~\ref{prop:d0=0} and by linearity, it means that $c'_{|\Pi_\infty} \in \PLift_q(m-1, k)$.
\end{proof}

\begin{remark}
  For $m=1$, we know that by definition, $\Lift_q(1, k) = \RS_q(k)$ and $\PLift_q(1, k) = \PRS_q(k)$. Therefore, Theorem~\ref{theo:puncturing-shortening} rewrites the well-known result stating that the shortening at the infinity of the projective Reed-Solomon code is a (classical) Reed-Solomon code.
\end{remark}

Theorem~\ref{theo:puncturing-shortening} also translates in terms of exact sequences:

\begin{corollary}
  The following exact sequence holds for every $1 \le k \le q-1$ and $m \ge 1$:
\[
0 \to \Lift_q(m, k-1) \to \PLift_q(m, k) \xrightarrow{\pi} \PLift_q(m-1,k) \to 0\,,
\]
where $\pi$ is the restriction map to points at infinity.
\end{corollary}

% =============================================================================
\section{On the practicality of projective lifted codes}
\label{sec:miscellaneous}

We here present miscellaneous results emphasizing the practicality of projective lifted codes. At first, we present tables and figures demonstrating the gain in terms of information rate, compared to projective Reed-Muller codes. In Subsection~\ref{subsec:automorphisms}, we prove that the storage cost of projective lifted codes can be reduced since they admit (quasi-)cyclic automorphisms. Explicit information sets are then computed in Subsection~\ref{subsec:information-sets}. We conclude this section by estimating the minimum distance (Subsection~\ref{subsec:distance}) and connecting our construction to a well-known family of design-based codes (Subsection~\ref{subsec:connection-design}).

% -----------------------------------------------------------------------------
\subsection{Information rate}

In this section, we emphasize how projective lifted codes surpasses projective Reed-Muller codes in terms of code rate (the local correcting capability being fixed). In Figure~\ref{fig:rates}, we present the rate of $\PRM_q(m, k)$ and $\PLift_q(m, k)$ for increasing values of $q = 2^e$. These codes are comparable since they have same length $n = \frac{q^{m+1}-1}{q-1}$, and same local correction features (locality and error tolerance). In each subfigure of Figure~\ref{fig:rates}, four curves are plotted: blue ones represent projective lifted codes and red ones projective Reed-Muller codes. Plain curves correspond to the minimum error tolerance setting, for which local correction admits no error on the line being picked (see Section~\ref{sec:local-correction}). To compare, dotted curves correspond to a constant fraction of errors tolerated by the local correcting algorithm. Here, the constant has been arbitrarily fixed to $1/32$.

\begin{figure}[h!]
  \centering  
  \includegraphics[scale=0.39]{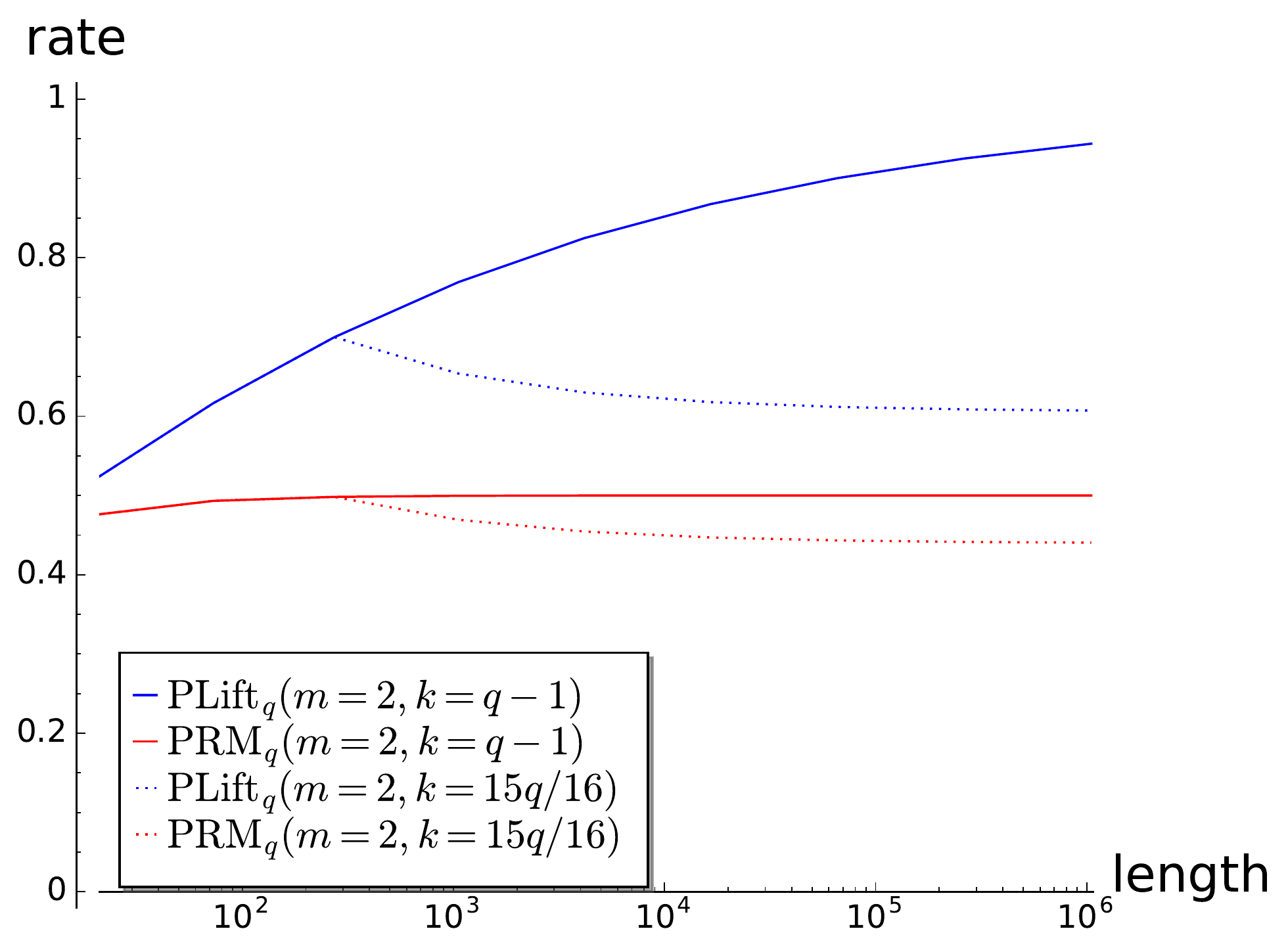}
  \includegraphics[scale=0.39]{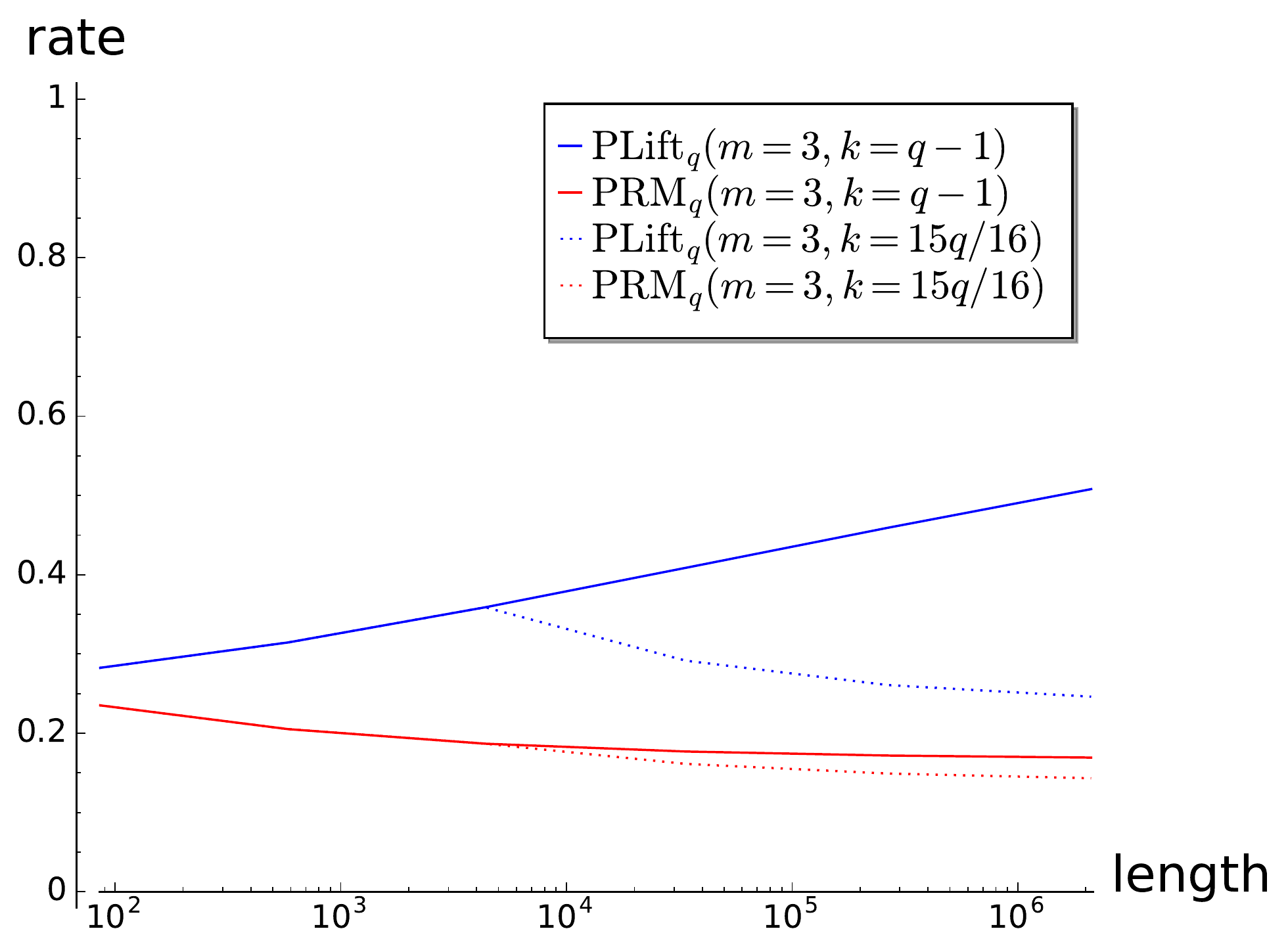}
  \includegraphics[scale=0.39]{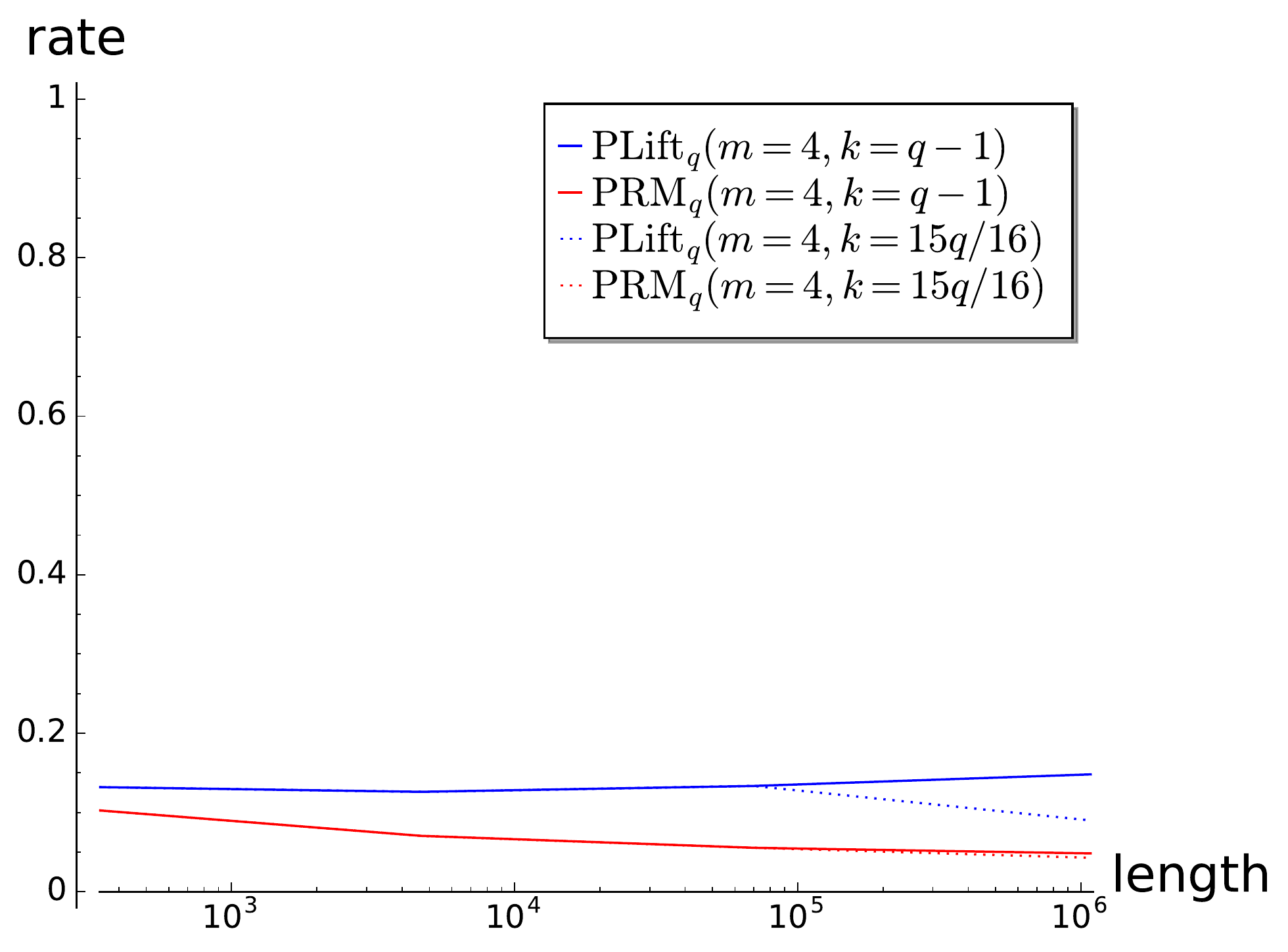}
  \caption{\label{fig:rates}Rate of projective Reed-Muller codes (red) and projective lifted codes (blue).}
\end{figure}

% -----------------------------------------------------------------------------
\subsection{Automorphisms and (quasi-)cyclicity}
\label{subsec:automorphisms}

In this section, we address the question of the (quasi-)cyclicity of projective lifted codes. More precisely, we prove in Proposition~\ref{prop:quasi-cyclic} that, under arithmetic constraints between $q$ and $m$, the code $\PLift_q(m, k)$ is a quasi-cyclic code up to diagonal isomorphims. This result relies deeply on the fact that $\PLift_q(m, k)$ is invariant under the action of $\Proj(\FF_q, m)$,  that has been proved in Lemma~\ref{lem:glm-stab-plift}.

In coding theory, automorphism groups of codes, and \emph{a fortiori} their permutation groups, are interesting for many reasons. For instance, they can be used for reducing the practical storage cost of the codes (through the storage of their generator or parity-check matrix). Cyclic or quasi-cyclic codes are known to be specifically efficient in that sense.

\begin{definition}[Cyclicity, quasi-cyclicity]
  A code $\calC \subseteq \FF_q^X$, $|X| = n$, is said \emph{cyclic} if $\mathrm{Perm}(\calC)$ contains a cyclic permutation of order $n$ (that is, an $n$-cycle). It is said \emph{quasi-cyclic} of index $c$ if $\mathrm{Perm}(\calC)$ contains a permutation which is the product of $c$ different $(n/c)$-cycles with disjoint orbits. In particular, a cyclic code is a quasi-cyclic code of index $1$.
\end{definition}

In all what follows, we fix a finite field $\FF_q$ and an integer $m \ge 1$, and we define $n = |\PP^m|$ and $d = \operatorname{gcd}(n, q-1)$.

\begin{definition}[representation of $\PP^m$]
  A tuple $\mathbf{u} = (\mathbf{u}_1, \dots, \mathbf{u}_n) \in (\FF_q^{m+1})^n$ represents $\PP^m$ if $\{ \mathbf{u}_1, \dots, \mathbf{u}_n \} = \PP^m$, when the $\mathbf{u}_i$ are taken up to projective equivalence.
\end{definition}

Let now $\phi : \FF_{q^{m+1}} \to \FF_q^{m+1}$ be an isomorphism of $\FF_q$-vector spaces, and $\omega$ be a primitive element of $\FF_{q^{m+1}}$. We define $\beta \mydef \omega^{q-1}$. It is clear that $\beta$ has order $n$ in the multiplicative group $\FF_{q^{m+1}}^{\times}$ since $q-1)n = q^{m+1}-1$. For every $0 \le i < d$, we define:
\[
U_i = \left( \phi\big(\omega^i \beta^d\big), \dots, \phi\big(\omega^i (\beta^d)^{n/d}\big) \right) \in (\FF_q^{m+1})^{n/d}\,.
\]
We also define its concatenation $U = U_0 \mid \hdots \mid U_{d-1} \in (\FF_q^{m+1})^n$.

\begin{lemma}
  If $n/d$ and $q-1$ are coprime, then $U$ represents $\PP^m$.
\end{lemma}
\begin{proof}
  We need to prove that all $\phi(\omega^i \beta^{dj})$ define distinct projective points for $0 \le i < d$ and $1 \le j \le n/d$. Since $\phi$ is bijective, it reduces to prove that, for $0 \le i_1, i_2 < d$ and $1 \le j_1,j_2 \le n/d$, if $\omega^{i_1 - i_2} \beta^{d(j_1 - j_2)} \in \FF_q$, then $(i_1,j_1) = (i_2,j_2)$.

Assume $(\omega^{i_1 - i_2} \beta^{d(j_1 - j_2)})^{q-1} = 1$. Then $\operatorname{ord}(\omega) = (q-1)n$ divides $(q-1) \times ((i_1 - i_2) + d(q-1)(j_1-j_2))$, that is, $n \mid (i_1 - i_2) + d(q-1)(j_1-j_2)$.

Since $d \mid n$, we get $d \mid (i_1 - i_2)$ which implies $i_1 = i_2$ because $0 \le i_1, i_2 < d$. Hence $n \mid d(q-1)(j_1-j_2)$, and our assumption $\operatorname{gcd}(n/d, q-1) = 1$ ensures that $(n/d) \mid j_1 - j_2$. Since $1 \le j_1, j_2 < n/d$, we finally obtain $j_1 = j_2$.
\end{proof}

Of course, every $\mathbf{u} = \phi(\omega^i \beta^{dj}) \in U_i$ is not necessarily represented in a standard form. Denote by $\mathbf{P_u} \in \FF_q^{m+1}$ its standard form. We have $\mathbf{u} = w_{\mathbf{u}} \mathbf{P_u}$ and we can define $w = (w_{\mathbf{u}})_{\mathbf{u} \in U} \in \FF_q^n$. Up to a reordering, if $n/d$ and $q-1$ are coprime, then we have
\[
U = w \star \calP\,
\]
where $\calP \in (\FF_q^{m+1})^n$ denotes the standard evaluation points of $\PP^m$. Similarly to the definition of $\ev_{\PP^m}$ given in the introduction, we can define a map $\ev_U : f \mapsto (f(\mathbf{u}) : \mathbf{u} \in U)$.

\begin{lemma}
  \label{lem:starC}
  Assume $n/d$ and $q-1$ are coprime. Let $\calC = \PLift_q(m,k)$, $v = k + (m-1)(q-1)$, and denote by $D = \Deg(\calC)$. Let also $\calC' = \ev_U(\Poly(D))$. Then, up to a permutation of coordinates,
  \[
  \calC' = w^v \star \calC,\,
  \]
  where $w^v$ denotes the $v$-fold $\star$-product of $w$ by itself. 
\end{lemma}

\begin{proof}
  If $\mathbf{P} \in \PP^m$ is represented by $\mathbf{u}$ in $U$, then by definition $(\ev_U(f))_{\mathbf{u}} = w_{\bf u}^v \ev_{\mathbf{P}}(f)$. Since $\calC' = \ev_U(\Poly(\Deg(\calC)))$, we get our result.
\end{proof}

Let us now introduce $\sigma : \FF_{q^{m+1}} \to \FF_{q^{m+1}}$ given by $x \mapsto \beta x$. We also denote by ${\psi \mydef \phi \circ \sigma \circ \phi^{-1}}$ the associated map over the vector space $\FF_q^{m+1}$. It is clear that $\psi \in \mathrm{Hom}(\FF_q^{m+1}, \FF_q^{m+1})$, and since $\sigma$ and $\phi$ are bijective, $\psi \in \mathrm{GL}_{m+1}(\FF_q)$. We finally denote by $\psi^i$ the $i$-fold composition of $\psi$. We then have $\psi^i(\mathbf{P}) = \phi(\omega^{i(q-1)}\phi^{-1}(\mathbf{P}))$ for any point $\mathbf{P} \in \PP^m$.

\begin{proposition}
  \label{prop:quasi-cyclic}
  If $n/d$ and $(q-1)$ are coprime, then $\calC' = w^v \star \PLift_q(m, k)$ is quasi-cyclic of index $d$, through the permutation $\psi^d \in \frak{S}(U)$. The orbits of $\psi^d$ are given by the subsets $U_i$.
\end{proposition}

\begin{proof}
  We can check that $\psi^d(U) = U$, hence $\psi^d \in \frak{S}(U)$. Since $\psi^d \in \GL_{m+1}(\FF_q)$, the polynomial space $\Poly(D)$ is invariant under $\psi^d$, where $D = \Deg(\calC)$. Besides, $\calC' = \ev_U(\Poly(D))$ thanks to Lemma~\ref{lem:starC}. Therefore $\psi^d \in \mathrm{Perm}(\calC')$.

Let us now prove that $\psi^d$ is an $(n/d)$-cycle. For $\phi(\omega^i \beta^{jd}) \in U_i$, we have 
\[
\psi^d(\phi(\omega^i \beta^{jd})) = \phi(\omega^i \beta^{jd} \beta^d) = \phi(\omega^i \beta^{(j+1)d}) \in U_i\,.
\]
It remains to show that the order of $\psi^d$ is $n/d$. Since $\phi$ is bijective and $U$ represents $\PP^m$, for every $0 \le s \le t < n/d$ we have:
\[
\forall \mathbf{u} \in U_i, (\psi^d)^s(\mathbf{u}) = (\psi^d)^t(\mathbf{u}) \iff \omega^{(t-s)d(q-1)} = 1 \iff n \mid (t-s)d(q-1)\,. 
\]
Our assumption on $n/d$ and $(q-1)$ implies that $t=s$; hence $\psi^d$ has order $n/d$.
\end{proof}

As an easy corollary, when $d = 1$ we obtain 
\begin{corollary}
  If $n$ and $q-1$ are coprime, then  for all $1 \le k \le q-1$ the code $w^v \star \PLift_q(m,k)$ is cyclic.
\end{corollary}

\begin{remark}
  A very similar approach was used by Berger and de Maximy in~\cite{BergerM01}, in order to prove the quasi-cyclicity of codes isomorphic to our definition of projective Reed-Muller codes.
\end{remark}

% -----------------------------------------------------------------------------
\subsection{Explicit information sets}
\label{subsec:information-sets}

In this section, we aim at giving explicit information sets for projective lifted codes. Such sets are useful in order to extend the local correctability of lifted codes to a local \emph{decodability} property (see~\cite{Yekhanin12}).

Our techniques are highly inspired by the work of Guo and Kopparty~\cite[Appendix A]{GuoK16}. We also prove a quite stronger result, being that a quite large family of affine evaluation codes presents the same information sets as affine lifted codes.

\paragraph{Monomiality of bounded degree affine evaluation codes.} Similarly to the previous section, let $\phi : \FF_{q^m} \to \FF_q^m$ be an $\FF_q$-isomorphism. We denote by $\FF_q[\mathbf{X}]^{\infty}_{q-1} \mydef \mathrm{Poly}(B^m_\infty(q-1))$ the space of $m$-variate polynomials of partial degree bounded by $q-1$. If $f \in \FF_q[\mathbf{X}]^{\infty}_{q-1}$ is seen as a function, then the map $f \circ \phi : \FF_{q^m} \to \FF_{q^m}$ can be interpolated uniquely as a univariate polynomial in $\FF_{q^m}[X]_{q^m-1}$. We denote by $\phi^*$ this process, which also appears to be an $\FF_q$-isomorphism:
\[
\begin{array}{rclc}
  \phi^* : &\FF_q[\mathbf{X}]^{\infty}_{q-1} &\to      &\FF_{q^m}[X]_{q^m-1}\\
           & f(\mathbf{X})                &\mapsto & (f \circ \phi)(X)
\end{array}
\]

We know that $\Omega \in \GL_m(\FF_q)$ acts on $m$-variate polynomials by $(\Omega, f(\mathbf{X})) \mapsto f(\Omega(\mathbf{X}))$. For some subspace $\calF$ of polynomials, we say that $\Omega \in \mathrm{Aut}(\calF)$ if $ \{ f \circ \Omega, f \in \calF \} \subseteq \calF$.

For a nonzero $a \in \FF_{q^m}^{\times}$, we denote by $\mu_a : \FF_{q^m} \to \FF_{q^m}$, $x \mapsto ax$. It is well-known that $\GL_1(\FF_{q^m}) = \{ \mu_a, a \in \FF_{q^m}^{\times} \}$. Every map $\mu_a$ being $\FF_q$-linear, we have $M_a \mydef \phi \circ \mu_a \circ \phi^{-1} \in \GL_m(\FF_q)$. Map $M_a$ is known as the \emph{$\FF_q$-homomorphism of the multiplication by $a \in \FF_{q^m}$}.

\begin{lemma}
  \label{lem:autom-phi*}
  Let $\calF$ be a subspace of $\FF_q[\mathbf{X}]_{q-1}^{\infty}$. If $\mathrm{GL}_m(\FF_q) \subseteq \mathrm{Aut}(\mathcal{F})$, then $\mathrm{GL}_1(\FF_{q^m}) \subseteq \mathrm{Aut}(\phi^*(\mathcal{F}))$.
\end{lemma}
\begin{proof}
  Let $f \circ \phi \in \phi^*(\mathcal{F})$. For every $\mu_a \in \mathrm{GL}_1(\FF_{q^m})$, we have $f \circ \phi \circ \mu_a = f \circ M_a \circ \phi$ by definition of the matrix of the multiplication by $a$. But $M_a \in \mathrm{GL}_m(\FF_q)$, hence $f \circ M_a \in \mathcal{F}$ and we get $f \circ \phi \circ \mu_a \in \phi^*(\mathcal{F})$.
\end{proof}

Let us define the subgroup of diagonal isomorphisms 
\[
\Diag(\FF_q, m) \mydef \{ \mathrm{Diag}_{\mathbf{a}} : \mathbf{P} \mapsto (a_1P_1,\dots,a_mP_m),\, \mathbf{a} \in (\FF_q^{\times})^m \} \subseteq \GL_m(\FF_q)\,.
\]

\begin{proposition}
  \label{prop:monomial-group}
  Let $\mathcal{F}$ be a subspace of $m$-variate polynomials of partial degree bounded by $q-2$, that is $\calF \subseteq \mathrm{Poly}(B^m_\infty(q-2))$. If $\Diag(\FF_q, m) \subseteq \mathrm{Aut}(\mathcal{F})$, then $\mathcal{F}$ is generated by monomials.
\end{proposition}
\begin{proof}
  Let $f \in \mathcal{F}$, such that $f(\mathbf{X}) = \sum_{\mathbf{d} \in D} f_{\mathbf{d}} \mathbf{X^d}$ with $D = \{ \mathbf{d} \in \NN^m, f_{\mathbf{d}} \ne 0 \}$. It is sufficient to prove that for all $\mathbf{d} \in D$, $\mathbf{X^d}$ lies in $\mathcal{F}$.

  Let $\mathbf{d} \in D$. Similarly to the proof of Theorem~\ref{thm:glm-invariant-implies-monomial}, we define
  \[
  Q_{\mathbf{d}}(\mathbf{X}) \mydef (-1)^m \sum_{\mathbf{a} \in (\FF_q^\times)^m} \Big(\prod_{i=1}^m a_i^{-d_i}\Big) (f \circ \mathrm{Diag}_{\mathbf{a}})(\mathbf{X})
  \]
  Since $\mathcal{F}$ is a vector space and $\Diag(\FF_q, m) \subseteq \mathrm{Aut}(\mathcal{F})$, we see that $Q_{\mathbf{d}}(\mathbf{X}) \in \mathcal{F}$. 
  \[
  \begin{aligned}
    Q_{\mathbf{d}}(\mathbf{X}) &= \sum_{\mathbf{a} \in (\FF_q^\times)^m} \Big(\prod_{i=1}^m -a_i^{-d_i}\Big) \sum_{\mathbf{j}} f_{\mathbf{j}} \,a_1^{j_1} \dots a_m^{j_m} \,\mathbf{X^j}\\
    &= \sum_{\mathbf{j}} f_{\mathbf{j}} \sum_{\mathbf{a} \in (\FF_q^{\times})^m} \Big(\prod_{i=1}^m - a_i^{j_i-d_i}\Big) \mathbf{X}^{\mathbf{j}}\\
   &= \sum_{\mathbf{j}} \, f_{\mathbf{j}} \prod_{i=1}^m \Big( \underbrace{- \sum_{a_i \in \FF_q^{\times}} a_i^{j_i-d_i}}_{= 0 \text{ if } d_i \ne j_i, \, 1 \text{ otherwise }}  \Big) \mathbf{X}^{\mathbf{j}} \quad=\quad f_{\mathbf{d}}\mathbf{X^d}\,.
  \end{aligned}
  \]

  We know that $\mathbf{d} \in D$, hence $f_{\mathbf{d}} \ne 0$ and by linearity we obtain $\mathbf{X^d} = \frac{1}{f_{\mathbf{d}}} Q_{\mathbf{d}}(\mathbf{X}) \in \mathcal{F}$.
\end{proof}

\paragraph{Information sets of some affine evaluation codes.} We first recall the definition of an information set of a linear code.

\begin{definition}[information set]
  Let $\calC \subseteq \FF_q^X$ be a linear code of dimension $k$ and support $X$, where $|X| = n$. An information set for $\calC$ is a subset $S \subseteq X$, $|S| = k$ such that the restriction of $\calC$ to coordinates in $X$ is $\FF_q^k$. In other words, $S$ is such that the projection of $\calC$ on $\FF_q^S$ is injective.
\end{definition}

\begin{lemma}
  \label{lem:pullback-information-sets}
  Let $\calF \subseteq \Poly(B^m_{\infty}(q-1))$ and assume that $S \subset \AA^1(\FF_{q^m})$ is an information set for $\ev_{\AA^1}(\phi^*(\mathcal{F}))$. Then $\phi(S)$ is an information set for $\ev_{\AA^m}(\mathcal{F})$.
\end{lemma}
\begin{proof}
  This follows from the fact that $\phi^*(\calF) = \{ f \circ \phi, f \in \calF \}$ and $\phi$ is an $\FF_q$-isomorphism.
\end{proof}

In the next proposition, we give a result that improves the theorem given by Guo and Kopparty in~\cite[Appendix A]{GuoK16}, in the specific case of codes evaluating polynomials with partial degree bounded by $q-2$ (which is the case for many interesting codes). Indeed, their result holds for affine-invariant codes while we only need codes invariant under $\GL_m(\FF_q)$.

\begin{proposition}
  Let $\mathcal{C} = \ev_{\AA^m}(\mathcal{F})$ be an affine evaluation code of dimension $k$ over $\FF_q$, and assume that $\calF \subseteq \Poly(B^m_\infty(q-2))$ and $\mathrm{GL}_m(\FF_q) \subseteq \mathrm{Aut}(\mathcal{F})$. Then, for every primitive element $\omega$ of $\FF_{q^m}$, and every isomorphism  $\phi : \FF_{q^m} \to \FF_q^m$, the set $\{ \phi(\omega), \dots, \phi(\omega^k) \}$ is an information set for $\calC$.
\end{proposition}
\begin{proof}
  The proof is highly inspired by~\cite[Appendix A]{GuoK16}. Thanks to Lemma~\ref{lem:pullback-information-sets}, it is sufficient to prove that $S = \{\omega, \dots, \omega^k\}$ is an information set for $\calC' = \ev_{\AA^1}(\phi^*(\calF))$. Moreover, since $\Diag(\FF_q, m) \subseteq \GL_m(\FF_q)$, the conjunction of Proposition~\ref{prop:monomial-group} and Lemma~\ref{lem:autom-phi*} ensures that $\calC'$ is monomial. Denote by $I = \Deg(\calC') = \{i_1,\dots,i_k\}$, and let $g(X) = \sum_{i \in I} a_i X^i \in \calF$. We need to prove:
\[
g \ne 0 \quad \implies \quad \ev_S(g) \ne 0\,.
\]
For this sake, we remark that
\[
\begin{pmatrix}
  \omega^{i_1}   & \omega^{i_2} & \dots  & \omega^{i_k} \\
  \omega^{2 i_1} & \omega^{2i_2} & \dots  & \omega^{2 i_k}\\
  \vdots       & \vdots       & \ddots & \vdots\\
  \omega^{k i_1} & \omega^{k i_2} & \dots & \omega^{k i_k}
\end{pmatrix}
\begin{pmatrix}
  a_1    \\
  a_2    \\
  \vdots \\
  a_k 
\end{pmatrix}
=
\begin{pmatrix}
  g(\omega)   \\
  g(\omega^2) \\
  \vdots       \\
  g(\omega^k)
\end{pmatrix}
= \ev_S(g)\,.
\]
Since the left-hand square matrix is a Vandermonde matrix and $\omega$ is primitive, it is invertible and the result is proved.
\end{proof}

As a corollary we recover Guo and Kopparty's result, since $\mathrm{GL}_m(\FF_q)$ is a subgroup of the group of affine transformations.
\begin{corollary}[given in~\cite{GuoK16}]
  \label{coro:information-set-lift}
Let $\mathcal{C} = \Lift_q(m, k)$ for $k \le q-2$. Then, for every $\omega$ primitive element of $\FF_{q^m}$, and every $\phi$ isomorphism $\FF_{q^m} \to \FF_q^m$, the set $\{ \phi(\omega), \dots, \phi(\omega^{\dim \calC}) \}$ is an information set for $\calC$.
\end{corollary}

\paragraph{The case of projective evaluation codes.} We would like to prove a similar result for projective lifted codes. Unfortunately, one cannot define an isomorphism between $\PP^1(\FF_{q^m})$ and $\PP^m(\FF_q)$ since they do not have same cardinality. To solve this issue, our idea is to decompose $\PP^m(\FF_q)$ into affine parts, and to use recursively the links between projective and affine lifted codes we stated in previous sections.

Let $\PP^m(\FF_q) = \bigsqcup_{i=0}^m \AA^{m,i}(\FF_q)$, where
\[
\AA^{m,i}(\FF_q) \mydef \{ (0:\dots:0:1:x_1:\dots:x_i), (x_1, \dots, x_i) \in \AA^i(\FF_q) \}\,.
\]
Informally, $\AA^{m,i}$ is the affine part of the $i$-dimensional projective subspace at infinity of $\PP^m$.

\begin{theorem}
Let $\calC = \PLift_q(m,k)$ for $k \le q-1$. Then, for every $\omega_i$ primitive element of $\FF_{q^i}$, and every isomorphism $\phi_i : \FF_{q^i} \to \AA^{m,i}(\FF_q)$, the set
\[
S = \bigsqcup_{i=0}^m \{ \phi_i(\omega_i), \dots, \phi_i(\omega_i^{\dim \calC_i}) \}
\]
is an information set for $\calC$, where $\calC_i = \Lift_q(i,k-1)$ for $i>0$, and by convention, $\dim \calC_0 = 1$ and $\phi_0(\FF_{q^0}) \mydef \{ (0:\dots:0:1) \}$.
\end{theorem}
\begin{proof}
  We proceed by induction on $m$.

$\bullet$ \emph{Case $m = 1$}. Then $\calC = \PRS_q(k)$ which is an MDS code of dimension $k+1$, hence any $(k+1)$-subset of $\PP^1$ is an information set for $\calC$. In particular, $S = \{ (0 : 1) \} \cup \{ \phi_1(\omega_1), \dots, \phi_1(\omega_1^k) \}$ is one of them.

$\bullet$ \emph{Induction step.} Assume the result holds for step $m-1$. A basis of $\PLift_q(m, k)$ consists in evaluating monomials with exponents in $\PDeg_q(m, k)$. Thanks to Theorem~\ref{theo:bij-degree-sets}, we know that $\PDeg_q(m, k)$ is in bijection with $\ADeg_q(m, k-1) \sqcup \PDeg_q(m-1, k)$, where the bijection is given in the proof of the theorem. Hence, there exists a generator matrix of $\PLift_q(m, k)$ defined as follows:

\begin{tikzpicture}
  \matrix [matrix of math nodes,left delimiter=(,right delimiter=),row sep=0.5cm,column sep=0.5cm] (m) {
     ~ & ~ & ~ & ~ & ~ & ~ & ~ & ~ \\
     ~ & ~ & ~ & ~ & ~ & ~ & ~ & ~ \\
     ~ & ~ & ~ & ~ & ~ & ~ & ~ & ~ \\
     ~ & ~ & ~ & ~ & ~ & ~ & ~ & ~ \\
     ~ & ~ & ~ & ~ & ~ & ~ & ~ & ~ \\
   };
   \node[fit=(m-1-1)(m-3-5)]{\Large{$G_0$}};
   \node[fit=(m-1-6)(m-3-8)]{\Large{$0$}};
   \node[fit=(m-4-1)(m-5-5)]{\Large{$*$}};
   \node[fit=(m-4-6)(m-5-8)]{\Large{$G_1$}};
  
  \draw[dashed] ($0.5*(m-1-5.north east)+0.5*(m-1-6.north west)$) --
  ($0.5*(m-5-5.south east)+0.5*(m-5-6.south west)$);
  
  \draw[dashed] ($0.5*(m-3-1.south west)+0.5*(m-4-1.north west)$) --
  ($0.5*(m-3-8.south east)+0.5*(m-4-8.north east)$);
  
  \node[above=10pt of m-1-1] (top-1) {\scriptsize{$\cdots$}};
  \node[above=10pt of m-1-2] (top-2) {};
  \node[above=10pt of m-1-3] (top-3) {\scriptsize{$(1:x_1:\dots:x_m)$}};
  \node[above=10pt of m-1-4] (top-4) {};
  \node[above=10pt of m-1-5] (top-5) {\scriptsize{$\cdots$}};
  \node[above=10pt of m-1-6] (top-6) {};
  \node[above=10pt of m-1-7] (top-7) {\scriptsize{$(0:\dots)$}};
  \node[above=10pt of m-1-8] (top-8) {};
  
  \node[left=12pt of m-1-1] (left-1) {};
  \node[left=12pt of m-2-1] (left-2) {};
  \node[left=12pt of m-3-1] (left-3) {};
  \node[left=12pt of m-4-1] (left-4) {};
  \node[left=12pt of m-5-1] (left-5) {};
  
  \node[rectangle, above delimiter=\{] (del-top-1) at ($0.5*(top-1.north) + 0.5*(top-5.north)$){
    \tikz{\path (top-1.west) rectangle (top-5.east);}
  };
  \node[above=10pt] at (del-top-1.north) {$\AA^m(\FF_q)$};

  \node[rectangle,above delimiter=\{] (del-top-2) at ($0.5*(top-6.north) + 0.5*(top-8.north)$) {
    \tikz{\path (top-6.west) rectangle (top-8.east);}
  };
  \node[above=10pt] at (del-top-2.north) {$\simeq \PP^{m-1}(\FF_q)$};
  
  \node[rectangle,left delimiter=\{] (del-left-1) at ($0.5*(left-1.east) +0.5*(left-3.east)$) {\tikz{\path (left-1.north) rectangle (left-3.south);}};
  \node[left=10pt,text width=5cm] at (del-left-1.west) {\small{evaluation of monomials with degrees in ${\ADeg(m,k-1)}$}};
  \node[rectangle,left delimiter=\{] (del-left-2) at ($0.5*(left-4.east) +0.5*(left-5.east)$) {\tikz{\path (left-4.north) rectangle (left-5.south);}};
  \node[left=10pt,text width=5cm] at (del-left-2.west) {\small{evaluation of monomials with degrees in ${\PDeg(m-1,k)}$}};
\end{tikzpicture}\\
where $G_0$ and $G_1$ are generator matrices of $\Lift_q(m,k-1)$ and $\PLift_q(m-1,k)$ respectively.

Since $G_1$ and $G_0$ are full-rank, we know that the union of an information set $S_0$ of $\Lift_q(m,k-1)$ and an information set $S_1$ of $\PLift_q(m-1,k)$ gives an information set $S$ of $\PLift_q(m,k)$. Information sets of affine lifted codes are described in Corollary~\ref{coro:information-set-lift} (we just need to take care about the way we represent affine points in the projective space, whence the definition of the $\AA^{m,i}$, $1 \le i \le m$). Therefore we have $S_0 = \{ \phi_m(\omega_m),\dots,\phi_m(\omega_m^{\dim \calC_m}) \}$ with $\phi_m$, $\omega_m$ defined as in the statement of the theorem. Besides, the inductive step gives the information set of $\PLift_q(m-1,k)$: $S_1 = \bigsqcup_{i=0}^{m-1} \{ \phi_i(\omega_i), \dots, \phi_i(\omega_i^{\dim \calC_i}) \}$.

Therefore $S = S_0 \sqcup S_1$ leads to the result at step $m$.
\end{proof}

% -----------------------------------------------------------------------------
\subsection{Estimation of the minimum distance}
\label{subsec:distance}

We give bounds on the minimum distance of a projective lifted code, depending on the minimum distance of the underlying projective Reed-Solomon code. In this section, $\mathrm{wt}(c)$ denotes the Hamming weight of a vector $c$, and $\mathrm{nz}_S(f)$ denotes the number of zeroes of $f \in \FF_q[X_0, \dots, X_m]^H_v$  over the set $S \subseteq \PP^m$.

\begin{proposition}[upper bound]
Let $1 \le k \le q-1$ and $\PRS_q(k)$ be the projective Reed-Solomon code of dimension $k+1$ and distance $d = q+1-k$. Then the distance $D$ of $\PLift_q(m, k)$ satisfies:
\[
D \le \theta_{m,q} - q^{m-1}(q+1 - d)
\]
where $\theta_{m,q} = \frac{q^{m+1}-1}{q-1}$. As a corollary, the relative distance $\delta$ of $\PRS_q(k)$ and $\Delta$ of $\PLift_q(m, k)$ satisfy:
\[
\Delta \le (1-b) \delta + b,\quad\text{ where } 0 \le b \le q^{-2}\,.
\]
\end{proposition}

\begin{proof}
  Let $c = \ev_{\PP^1}(g) \in \PRS_q(k)$ be a minimum-weight codeword, \emph{i.e.} $\mathrm{wt}(c) = d$. Assume that $g(X_0, X_1) = \sum_{i=0}^k g_i X_0^i X_1^{k-i}$, and let
\[
f(X_0,\dots,X_m) \mydef g_0 X_1^{(q-1)(m-1) + k} + \sum_{i=1}^k g_i X_0^{(q-1)(m-1) + i} X_1^{k-i} \in \FF_q[X_0,\dots,X_m]^H_v
\]
where $v = (q-1)(m-1) + k$. By studying the degrees of $f$, one can check that $c' \mydef \ev_{\PP^m}(f) \in \PLift_q(m, k)$. Moreover, for every $(x_0 : x_1) \in \PP^1$, we have:
\[
f(x_0, x_1, x_2, \dots,x_m) = g(x_0, x_1),\quad  \forall \mathbf{x} = (x_2,\dots,x_m) \in \FF_q^{m-1}\,. 
\]
Hence $c'$ is non-zero, and:
\[
  D \le \mathrm{wt}(c') = \theta_{m,q} - \mathrm{nz}_{\PP^m}(f) \le \theta_{m,q} - \mathrm{nz}_{\PP^1}(g)\,q^{m-1} =  \theta_{m,q} - q^{m-1} (q+1 - d)\,.
\]
For the bound on the relative distance, we divide both sides of the previous equation by $\theta_{m,q}$ and we use that $d = (q+1)\delta$ by definition. Then we get:
\[
\Delta \le 1 +  (\delta - 1) a\,,
\]
where $a = \frac{(q+1)q^{m-1}}{\theta_{m,q}} = 1 - \frac{q^{m-1}-1}{q^{m+1}-1}$ satisfies $1 - q^{-2} \le a \le 1$. Denoting $b = 1-a$ concludes the proof.
\end{proof}

\begin{proposition}[lower bound]
  Let $1 \le k \le q-1$ and $\PRS_q(k)$ be a projective Reed-Solomon code of dimension $k+1$ and distance $d = q+1-k$. Then the distance $D$ of $\PLift_q(m, k)$ satisfies:
\[
D \ge (d-1) \theta_{m-1,q} + 1\,
\]
where $\theta_{m,q} = \frac{q^{m+1}-1}{q-1}$. As a corollary, the respective relative distance $\delta$  of $\PRS_q(k)$ and $\Delta$ of $\PLift_q(m, k)$ satisfy:
\[
\Delta \ge (1-b') \delta - b',\quad  \text{ where } 0 \le b' \le q^{-1}\,.
\]
\end{proposition}

\begin{proof}
  Let $c = \ev_{\PP^m}(f) \in \PLift_q(m, k)$ be a minimum-weight codeword, meaning that $D = \mathrm{wt}(c) \ne 0$. Let also $\mathbf{a} \in \PP^m$ such that $c_{\mathbf{a}} = \ev_{\mathbf{a}}(f) \ne 0$. We denote by $\Lambda_{\mathbf{a}}$ the set of projective lines of $\PP^m$ passing through $\mathbf{a}$. It is clear that $( \bigcup_{L \in \Lambda_{\mathbf{a}}} (L \setminus \{\mathbf{a}\}) ) \cup \{\mathbf{a} \}$ is a partition of $\PP^m$, and $|\Lambda_{\mathbf{a}}| = \theta_{m-1,q} = \frac{q^m-1}{q-1}$. Besides we have:
\[
  \sum_{L \in \Lambda_{\mathbf{a}}} \mathrm{wt}(c_{|L}) = |\theta_{m-1,q}| +  \sum_{L \in \Lambda_{\mathbf{a}}} \mathrm{wt}(c_{|L \setminus \{\mathbf{a}\}}) = \theta_{m-1,q} + \mathrm{wt}(c_{|\PP^m \setminus \{\mathbf{a}\}}) = \theta_{m-1,q} + (D-1)\,.
\]
Therefore, there must exist a line $L_0$ such that
\[
\mathrm{wt}(c_{|L_0}) \le \frac{1}{|\Lambda_{\mathbf{a}}|} \sum_{L \in \Lambda_{\mathbf{a}}} \mathrm{wt}(c_{|L}) = 1 + \frac{D-1}{\theta_{m-1,q}}\,.
\]
Since $c_{|L_0} \in \PRS_q(k)$ and $c_{|L_0} \ne 0$, its weight is greater than $d$ and we get:
\[
D \ge (d-1) \theta_{m-1,q} + 1\,.
\]
Dividing both sides by $\theta_{m,q}$ and using $\theta_{m,q} = q\theta_{m-1,q} + 1$ finally leads to:
\[
\Delta \ge \frac{(q+1)\theta_{m-1,q}}{\theta_{m,q}} \delta - \frac{\theta_{m-1,q}-1}{\theta_{m,q}} \ge (1-b')\delta - b'\,,
\]
where $b' = \frac{\theta_{m-1,q}-1}{\theta_{m,q}} = qb$ and $b$ is defined in the previous proposition.
\end{proof}

% -----------------------------------------------------------------------------
\subsection{Connection with codes based on projective geometry designs}
\label{subsec:connection-design}

In this section, we simply point out a link between the construction of lifted codes and the codes coming from design theory --- we refer to~\cite{AssmusK92} as a good reference for links between codes and designs. We focus on projective lifted codes since they are the core of our work, but the upcoming facts also hold for affine lifted codes.

Let us consider the highest value of $k$ for which $\PLift_q(m,k)$ is non-trivial, that is $k = q-1$. It is well-known that dual codes of projective Reed-Solomon codes are also projective Reed-Solomon codes, and in the setting $k = q-1$ we have:
\begin{lemma}
  The dual code of $\PRS_q(q-1)$ is the repetition code $\langle (1,\dots,1) \rangle = \PRS_q(0)$ of length $q+1$ over $\FF_q$.
\end{lemma}

Hence, a (non full-rank) parity-check matrix $H$ for $\PLift_q(m,k)$ can be written by listing in rows the incidence vectors of lines and points of the projective space $\PP^m$. More formally,
\[
H = \begin{pmatrix} ~ & \mathds{1}_{L_1} & ~ \\ ~ &  \vdots  & ~  \\  ~ &  \mathds{1}_{L_N} & ~  \end{pmatrix}\,
\]
where $\{L_1,\dots,L_N\}$ denotes the set of all the projective lines of $\PP^m$, and $\mathds{1}_X$ is the $\{0,1\}$-vector of length $\theta_{m,q} = |\PP^m|$ which is $1$ at coordinate $i$ if and only if $i \in X$ (for any $X \subset \PP^m$).

In fact, matrix $H$ is exactly the incidence matrix of the \emph{projective geometry design} $\mathrm{PG}_1(m, q)$, the block design of points and lines in the projective space $\PP^m$. Moreover, the vector space over $\FF_q$ spanned by this matrix gives rise to a linear code, which has been thoroughly studied and whose significant properties are given in~\cite{AssmusK92}. This code is known as the \emph{code spanned by the design} $\mathrm{PG}_1(m, q)$, and is denoted by $\calC(\mathrm{PG}_1(m, q))$. To sum up we have:

\begin{lemma}
  For every prime power $q$ and every $m \ge 2$, the projective lifted code $\PLift_q(m, q-1)$ and the code $\calC (\mathrm{PG}_1(m, q))$ spanned by the projective geometry design $\mathrm{PG}_1(m, q)$ are dual codes.
\end{lemma}

This characterisation allows us to obtain the dimension of projective lifted codes, for which the rank of matrices $H$ has been computed. For instance, it is proved (\emph{e.g.} in \cite{Smith69}) that the rank over $\FF_{p^t}$ of the design of points and lines in $\PP^2(\FF_{p^t})$ is $\binom{1+p}{2}^t + 1$. Therefore,

\begin{corollary}
For any $t\ge 1$ and any prime $p$, we have:
\[
  \dim \left( \PLift_{p^t}(2, p^t-1) \right) = p^{2t} + p^t - \Big( \frac{p(p+1)}{2} \Big)^t\,.
\]
\end{corollary}

% =============================================================================
\section{Conclusion}

In this work we introduced lifted projective Reed-Solomon codes as an analogue of the lifting of Reed-Solomon codes studied by Guo, Kopparty and Sudan in~\cite{GuoKS13}. We presented local correcting algorithms for these codes, and proved their practicality through explicit bases, information sets and automorphisms. However, similarly to the affine setting, we still lack closed formulae for the dimension of the codes. Future works may then consist in keeping studying the lifting process, for a better understanding of the structure of lifted codes. A generalisation of our work to the lifting of projective Reed-Muller codes or other codes invariant under $\Proj(\FF_q, t)$ would also be of interest.

% =============================================================================
\subsection*{Acknowledgements}
This work is partially funded by  French ANR-15-CE39-0013-01 \enquote{Manta}.
%The author is very grateful to Daniel Augot and Françoise Levy-dit-Vehel for their very helpful comments on the paper.
The author would like to thank Françoise Levy-dit-Vehel and Daniel Augot for their valuable comments and advice concerning the presentation of the results.
% \end{acknowledgements}

% =============================================================================
\bibliographystyle{alpha}
\bibliography{projective-lift}

\newcommand{\etalchar}[1]{$^{#1}$}
\begin{thebibliography}{RTM79}

\bibitem[AK92]{AssmusK92}
Edward~F. Assmus and Jennifer~D. Key.
\newblock {\em Designs and Their Codes}.
\newblock Cambridge University Press, 1992.

\bibitem[Alo99]{Alon99}
Noga Alon.
\newblock Combinatorial nullstellensatz.
\newblock {\em Combinatorics, Probability and Computing}, 8(1-2):7--29, January
  1999.

\bibitem[BC93]{BergerC93}
Thierry~P. Berger and Pascale Charpin.
\newblock {The automorphism group of Generalized Reed-Muller codes}.
\newblock {\em Discrete Mathematics}, 117(1-3):1--17, 1993.

\bibitem[BdM01]{BergerM01}
Thierry~P. Berger and Louis de~Maximy.
\newblock {Cyclic Projective Reed-Muller Codes}.
\newblock In Serdar Boztas and Igor~E. Shparlinski, editors, {\em Applied
  Algebra, Algebraic Algorithms and Error-Correcting Codes, 14th International
  Symposium, AAECC-14, Melbourne, Australia November 26-30, 2001, Proceedings},
  volume 2227 of {\em Lecture Notes in Computer Science}, pages 77--81.
  Springer, 2001.

\bibitem[Ber02]{Berger02}
Thierry~P. Berger.
\newblock {Automorphism groups of homogeneous and projective Reed-Muller
  codes}.
\newblock {\em {IEEE} Trans. Information Theory}, 48(5):1035--1045, 2002.

\bibitem[GK16]{GuoK16}
Alan Guo and Swastik Kopparty.
\newblock List-decoding algorithms for lifted codes.
\newblock {\em {IEEE} Trans. Information Theory}, 62(5):2719--2725, 2016.

\bibitem[GKS13]{GuoKS13}
Alan Guo, Swastik Kopparty, and Madhu Sudan.
\newblock New affine-invariant codes from lifting.
\newblock In Robert~D. Kleinberg, editor, {\em Innovations in Theoretical
  Computer Science, {ITCS} '13, Berkeley, CA, USA, January 9-12, 2013}, pages
  529--540. {ACM}, 2013.

\bibitem[KR06]{KaufmanR06}
Tali Kaufman and Dana Ron.
\newblock Testing polynomials over general fields.
\newblock {\em {SIAM} J. Comput.}, 36(3):779--802, 2006.

\bibitem[KS08]{KaufmanS08}
Tali Kaufman and Madhu Sudan.
\newblock Algebraic property testing: the role of invariance.
\newblock In Cynthia Dwork, editor, {\em Proceedings of the 40th Annual {ACM}
  Symposium on Theory of Computing, Victoria, British Columbia, Canada, May
  17-20, 2008}, pages 403--412. {ACM}, 2008.

\bibitem[KSY14]{KoppartySY14}
Swastik Kopparty, Shubhangi Saraf, and Sergey Yekhanin.
\newblock High-rate codes with sublinear-time decoding.
\newblock {\em J. {ACM}}, 61(5):28:1--28:20, 2014.

\bibitem[Lac86]{Lachaud86}
Gilles Lachaud.
\newblock Projective {R}eed-{M}uller codes.
\newblock In {\em Coding Theory and Applications}, volume 311 of {\em Lecture
  Notes in Computer Science}, pages 125--129. Springer, 1986.

\bibitem[Lac90]{Lachaud90}
Gilles Lachaud.
\newblock The parameters of projective {R}eed-{M}{\"{u}}ller codes.
\newblock {\em Discrete Mathematics}, 81(2):217--221, 1990.

\bibitem[RTM79]{ReedTM79}
Irving~S. Reed, Trieu-Kien Truong, and Robert~L. Miller.
\newblock Simplified algorithm for correcting both errors and erasures of
  {R}eed-{S}olomon codes.
\newblock {\em Institution of Electrical Engineers}, 126:961--963, 1979.

\bibitem[RTR97]{RenteriaT97}
Carlos Renterìa and Horacio Tapia-Recillas.
\newblock {R}eed-{M}uller codes: an ideal theory approach.
\newblock {\em Communications in Algebra}, 25(2):401--413, 1997.

\bibitem[S{\etalchar{+}}17]{SageMath}
W.~A. Stein et~al.
\newblock {\em {S}age {M}athematics {S}oftware ({V}ersion 8.0)}.
\newblock The Sage Development Team, 2017.

\bibitem[Smi69]{Smith69}
K.J.C. Smith.
\newblock On the p-rank of the incidence matrix of points and hyperplanes in a
  finite projective geometry.
\newblock {\em Journal of Combinatorial Theory}, 7(2):122--129, 1969.

\bibitem[S{\o}r91]{Sorensen91}
Anders~Bj{\ae}rt S{\o}rensen.
\newblock Projective {R}eed-{M}uller codes.
\newblock {\em {IEEE} Trans. Information Theory}, 37(6):1567--1576, 1991.

\bibitem[Sti04]{Stinson04}
Douglas~R. Stinson.
\newblock {\em Combinatorial Designs - Constructions and Analysis}.
\newblock Springer, 2004.

\bibitem[Yek12]{Yekhanin12}
Sergey Yekhanin.
\newblock Locally decodable codes.
\newblock {\em Foundations and Trends in Theoretical Computer Science},
  6(3):139--255, 2012.

\end{thebibliography}

\appendix

% =============================================================================
\section{Useful results}
\label{app:useful-results}

% -----------------------------------------------------------------------------
\subsection{Combinatorial \emph{Nullstellensatz}}

We recall the Combinatorial \emph{Nullstellensatz} proved by Alon in~\cite{Alon99}.
\begin{theorem}[Combinatorial \emph{Nullstellensatz}~\cite{Alon99}]
  \label{thm:combinatorial-nullstellensatz}
  Let $\FF$ be a field and $f \in \FF[X_1, \dots, X_r]$. Assume that $\deg(f) = \sum_{i=1}^r t_i$ and the coefficient of the monomial $\mathbf{X^t} = \prod_{i=1}^r X_i^{t_i}$ in $f$ is non-zero (in other words, assume that $\mathbf{t} = (t_1, \dots, t_r)$ is a degree of $f$). Let finally $W_1, \dots, W_r \subseteq \FF$ such that $|W_i| > t_i$ for every $1 \le i \le r$.

Then, there exists $\mathbf{w} \in W_1 \times \dots \times W_r$ such that $f(\mathbf{w}) \ne 0$.
\end{theorem}

% -----------------------------------------------------------------------------
\subsection{Technical results}

\begin{lemma}
  \label{lem:technical-sum}
  The following equality over bivariate polynomials holds:
  \[
  \sum_{\beta \in \FF_q} (\beta X + Y)^{q-1} = -X^{q-1}\,.
  \]
\end{lemma}
\begin{proof}
  Let $\alpha$ be a primitive element of $\FF_q$.
  \[
  \begin{aligned}
  \sum_{\beta \in \FF_q} (\beta X + Y)^{q-1} &= Y^{q-1} + \sum_{i=0}^{q-2} (\alpha^i X + Y)^{q-1} = Y^{q-1} + \sum_{i=0}^{q-2} \sum_{j=0}^{q-1} \binom{q-1}{j} \alpha^{ij} X^j Y^{q-1-j}\\
  &= Y^{q-1} + \sum_{j=0}^{q-1} \binom{q-1}{j} \left( \sum_{i=0}^{q-2} (\alpha^j)^i \right) X^j Y^{q-1-j}\\
  &= Y^{q-1} + \left( \binom{q-1}{0} \times (-1) \times X^0 Y^{q-1} \right) + \left( \binom{q-1}{q-1} \times (-1) \times X^{q-1} Y^0 \right) \\
  &= -X^{q-1}
  \end{aligned}
  \]
\end{proof}

% -----------------------------------------------------------------------------
\subsection{Automorphism groups of (projective) Reed-Muller codes}
\label{app:automorphism-groups-codes}

The automorphism group of affine Reed-Muller codes has been thoroughly studied by Berger and Charpin in~\cite{BergerC93} with group algebra techniques. For our needs, we recall below that this group contains the subgroup of affine transformations. 

\begin{proposition}[Reed-Muller code]
  Let $0 \le k \le m(q-1)$. The automorphism group of the Reed-Muller code $\calC = \RM_q(m, k)$ contains the affine permutations $\Aff(\FF_q, m)$.
\end{proposition}
\begin{proof}
  Let $c = \ev_{\AA^m}(f) \in \calC$, $M \in \GL_m(\FF_q)$ and $\mathbf{b} \in \FF_q^m$. Denote by $T_{M, \mathbf{b}}(\mathbf{x}) = M\mathbf{x} + \mathbf{b}$ for every $\mathbf{x} \in \AA^m$. Let us prove that $T_{M, \mathbf{b}}(c) \in \calC$.

We remark that $T_{M, \mathbf{b}}(c) = \ev_{\AA^m}(f \circ T_{M, \mathbf{b}}^{-1})$. Since $T_{M, \mathbf{b}}$ is affine, so is $T_{M, \mathbf{b}}^{-1}$, and the total degree of $f \circ T_{M, \mathbf{b}}^{-1}$ is the same that the total degree of $f$. Hence $\ev_{\AA^m}(f \circ T_{M, \mathbf{b}}^{-1}) \in \RM_q(m, k)$ and the proof is completed.
\end{proof}

A few years later, Berger also studied the automorphism group of projective Reed-Muller codes~\cite{Berger02}.

\begin{proposition}[projective Reed-Muller code]
  \label{prop:automorphism-group-PRM}
  Let $0 \le v \le m(q-1)$. The automorphism group of the projective Reed-Muller code $\calC = \PRM_q(m, v)$ contains the projective automorphisms $\Proj(\FF_q, m)$.
\end{proposition}
\begin{proof}
  Using that
  \[
  \ev_{\PP^m}(f \circ M) = w_M^v \star \sigma_{M^{-1}}(\ev_{\PP^m}(f))
  \]
  for every $(w_M^v, \sigma_{M^{-1}}) \in \Proj(\FF_q, m)$, the proof is very similar to the previous one.
\end{proof}

% =============================================================================
\section{Building the query generator $\calR_s$}
\label{app:query-generator}

We recall that in our local correction algorithm (Section~\ref{subsec:local-correcting-algorithms}) we need a randomized query generator $\calR_s$ which, given a point $\mathbf{P} \in \PP^m$ and an embedding $L \in \Emb_\PP(m, \mathbf{P})$, returns $s$ random points of $L(\PP^1)$ such that:
\begin{equation}
  \label{eq:uniform-probability}
\forall \mathbf{Q} \in \PP^m, \mathrm{Pr}_{\substack{L \leftarrow \Emb_\PP(m, \mathbf{P})\\S \leftarrow \calR_s(\mathbf{P}, L)}} [ \mathbf{Q} \in S ] = s/n\,.
\end{equation}

The tricky point is that, for a fixed $L \in \Emb_\PP(m, \mathbf{P})$, we cannot pick the $s$ points \emph{uniformly} at random on $L(\PP^1)$, otherwise the point $\mathbf{P}$ will have a larger probability to be chosen than the other points. We provide a solution to this issue in Algorithm~\ref{algo:query-generator}.

\begin{algorithm}[h!]
  \KwData{a point $\mathbf{P}$, an embedding $L \in \Emb_\PP(m, \mathbf{P})$}
  \KwResult{$s$ points $S = \{\mathbf{Q}_1, \dots, \mathbf{Q}_s\} \subseteq L(\PP^1)$ satisfying~\eqref{eq:uniform-probability}}
  Set $S = \varnothing$.\\
  Add $\mathbf{P}$ to $S$ with probability $s/n$.\\
  \eIf{$\mathbf{P} \in S$}{
    Pick uniformly at $U$ random an $(s-1)$-subset of $L(\PP^1) \setminus \{\mathbf{P}\}$.\\
    $S := S \cup U$.\\
  }{
    
    Pick uniformly at $U$ random an $s$-subset of $L(\PP^1) \setminus \{\mathbf{P}\}$.\\
    $S := U$.\\
  }
  Return $S$.
  \caption{\label{algo:query-generator}Query generator $\calR_s$}
\end{algorithm}

% =============================================================================
\section{Computation of the dimension of lifted codes}
\label{app:computation}

In the following tables are presented some parameters of affine lifted codes, projective lifted codes and projective Reed-Muller codes. We denote respectively by $n_A$, $\dim(A)$ and $R_A$ the length, dimension and rate of $A = \Lift_q(m, k-1)$ the value of $k$ given in the first row ($q$ and $m$ being fixed in each table). Similarly, $n_P$, $\dim(P)$ and $R_P$ represent the length, dimension and rate of $P = \PLift_q(m, k)$, while $\dim(\PRM)$ and $R_{\PRM}$ denote the dimension and the rate of $\PRM_q(m, k)$ (its length being $n_P$).

We choose to compare these codes because, in the local correcting algorithm, they admit approximately the same error-correction capability and locality. Our goal is to show how lifting leads to higher rates, and that projective and affine lifted codes behave similarly.

\subsection{Parameters of the kind $m=2$, $q=2^t$}

\begin{table}[!h]
\centering
\begin{tabular}{c|ccc|ccc|cc}
$k$ & $n_A$ & $\dim(A)$ & $R_A$ & $n_P$ & $\dim(P)$ & $R_P$ & $\dim(\PRM)$ & $R_{\PRM}$ \\ \hline
$1$ & $16$ & $1$ & $0.0625$ & $21$ & $3$ & $0.143$ & $3$ & $0.143$ \\
$2$ & $16$ & $3$ & $0.188$ & $21$ & $6$ & $0.286$ & $6$ & $0.286$ \\
$3$ & $16$ & $7$ & $0.438$ & $21$ & $11$ & $0.524$ & $10$ & $0.476$ \\
\end{tabular}
\caption{Parameters for $m=2$ and $q=4$.}
\end{table}

\begin{table}[!h]
\centering
\begin{tabular}{c|ccc|ccc|cc}
$k$ & $n_A$ & $\dim(A)$ & $R_A$ & $n_P$ & $\dim(P)$ & $R_P$ & $\dim(\PRM)$ & $R_{\PRM}$ \\ \hline
$1$ & $64$ & $1$ & $0.0156$ & $73$ & $3$ & $0.0411$ & $3$ & $0.0411$ \\
$2$ & $64$ & $3$ & $0.0469$ & $73$ & $6$ & $0.0822$ & $6$ & $0.0822$ \\
$3$ & $64$ & $6$ & $0.0938$ & $73$ & $10$ & $0.137$ & $10$ & $0.137$ \\
$4$ & $64$ & $10$ & $0.156$ & $73$ & $15$ & $0.205$ & $15$ & $0.205$ \\
$5$ & $64$ & $16$ & $0.25$ & $73$ & $22$ & $0.301$ & $21$ & $0.288$ \\
$6$ & $64$ & $24$ & $0.375$ & $73$ & $31$ & $0.425$ & $28$ & $0.384$ \\
$7$ & $64$ & $37$ & $0.578$ & $73$ & $45$ & $0.616$ & $36$ & $0.493$ \\
\end{tabular}
\caption{Parameters for $m=2$ and $q=8$.}
\end{table}

\begin{table}[!h]
\centering
\begin{tabular}{c|ccc|ccc|cc}
$k$ & $n_A$ & $\dim(A)$ & $R_A$ & $n_P$ & $\dim(P)$ & $R_P$ & $\dim(\PRM)$ & $R_{\PRM}$ \\ \hline
$8$ & $256$ & $36$ & $0.141$ & $273$ & $45$ & $0.165$ & $45$ & $0.165$ \\
$9$ & $256$ & $46$ & $0.18$ & $273$ & $56$ & $0.205$ & $55$ & $0.201$ \\
$10$ & $256$ & $58$ & $0.227$ & $273$ & $69$ & $0.253$ & $66$ & $0.242$ \\
$11$ & $256$ & $72$ & $0.281$ & $273$ & $84$ & $0.308$ & $78$ & $0.286$ \\
$12$ & $256$ & $88$ & $0.344$ & $273$ & $101$ & $0.37$ & $91$ & $0.333$ \\
$13$ & $256$ & $109$ & $0.426$ & $273$ & $123$ & $0.451$ & $105$ & $0.385$ \\
$14$ & $256$ & $135$ & $0.527$ & $273$ & $150$ & $0.549$ & $120$ & $0.44$ \\
$15$ & $256$ & $175$ & $0.684$ & $273$ & $191$ & $0.7$ & $136$ & $0.498$ \\
\end{tabular}
\caption{Parameters for $m=2$ and $q=16$.}
\end{table}

\begin{table}[!h]
\centering
\begin{tabular}{c|ccc|ccc|cc}
$k$ & $n_A$ & $\dim(A)$ & $R_A$ & $n_P$ & $\dim(P)$ & $R_P$ & $\dim(\PRM)$ & $R_{\PRM}$ \\ \hline
$24$ & $1024$ & $336$ & $0.328$ & $1057$ & $361$ & $0.342$ & $325$ & $0.307$ \\
$25$ & $1024$ & $373$ & $0.364$ & $1057$ & $399$ & $0.377$ & $351$ & $0.332$ \\
$26$ & $1024$ & $415$ & $0.405$ & $1057$ & $442$ & $0.418$ & $378$ & $0.358$ \\
$27$ & $1024$ & $462$ & $0.451$ & $1057$ & $490$ & $0.464$ & $406$ & $0.384$ \\
$28$ & $1024$ & $514$ & $0.502$ & $1057$ & $543$ & $0.514$ & $435$ & $0.412$ \\
$29$ & $1024$ & $580$ & $0.566$ & $1057$ & $610$ & $0.577$ & $465$ & $0.44$ \\
$30$ & $1024$ & $660$ & $0.645$ & $1057$ & $691$ & $0.654$ & $496$ & $0.469$ \\
$31$ & $1024$ & $781$ & $0.763$ & $1057$ & $813$ & $0.769$ & $528$ & $0.5$ \\
\end{tabular}
\caption{Parameters for $m=2$ and $q=32$.}
\end{table}

\begin{table}[!h]
\centering
\begin{tabular}{c|ccc|ccc|cc}
$k$ & $n_A$ & $\dim(A)$ & $R_A$ & $n_P$ & $\dim(P)$ & $R_P$ & $\dim(\PRM)$ & $R_{\PRM}$ \\ \hline
$56$ & $4096$ & $2004$ & $0.489$ & $4161$ & $2061$ & $0.495$ & $1653$ & $0.397$ \\
$57$ & $4096$ & $2122$ & $0.518$ & $4161$ & $2180$ & $0.524$ & $1711$ & $0.411$ \\
$58$ & $4096$ & $2254$ & $0.55$ & $4161$ & $2313$ & $0.556$ & $1770$ & $0.425$ \\
$59$ & $4096$ & $2400$ & $0.586$ & $4161$ & $2460$ & $0.591$ & $1830$ & $0.44$ \\
$60$ & $4096$ & $2560$ & $0.625$ & $4161$ & $2621$ & $0.63$ & $1891$ & $0.454$ \\
$61$ & $4096$ & $2761$ & $0.674$ & $4161$ & $2823$ & $0.678$ & $1953$ & $0.469$ \\
$62$ & $4096$ & $3003$ & $0.733$ & $4161$ & $3066$ & $0.737$ & $2016$ & $0.484$ \\
$63$ & $4096$ & $3367$ & $0.822$ & $4161$ & $3431$ & $0.825$ & $2080$ & $0.5$ \\
\end{tabular}
\caption{Parameters for $m=2$ and $q=64$.}
\end{table}

\newpage
\subsection{Parameters of the kind $m=3$, $q=2^t$}

\begin{table}[!h]
\centering
\begin{tabular}{c|ccc|ccc|cc}
$k$ & $n_A$ & $\dim(A)$ & $R_A$ & $n_P$ & $\dim(P)$ & $R_P$ & $\dim(\PRM)$ & $R_{\PRM}$ \\ \hline
$1$ & $64$ & $1$ & $0.0156$ & $85$ & $4$ & $0.0471$ & $4$ & $0.0471$ \\
$2$ & $64$ & $4$ & $0.0625$ & $85$ & $10$ & $0.118$ & $10$ & $0.118$ \\
$3$ & $64$ & $13$ & $0.203$ & $85$ & $24$ & $0.282$ & $20$ & $0.235$ \\
\end{tabular}
\caption{Parameters for $m=3$ and $q=4$.}
\end{table}

\begin{table}[!h]
\centering
\begin{tabular}{c|ccc|ccc|cc}
$k$ & $n_A$ & $\dim(A)$ & $R_A$ & $n_P$ & $\dim(P)$ & $R_P$ & $\dim(\PRM)$ & $R_{\PRM}$ \\ \hline
$1$ & $512$ & $1$ & $0.00195$ & $585$ & $4$ & $0.00684$ & $4$ & $0.00684$ \\
$2$ & $512$ & $4$ & $0.00781$ & $585$ & $10$ & $0.0171$ & $10$ & $0.0171$ \\
$3$ & $512$ & $10$ & $0.0195$ & $585$ & $20$ & $0.0342$ & $20$ & $0.0342$ \\
$4$ & $512$ & $20$ & $0.0391$ & $585$ & $35$ & $0.0598$ & $35$ & $0.0598$ \\
$5$ & $512$ & $38$ & $0.0742$ & $585$ & $60$ & $0.103$ & $56$ & $0.0957$ \\
$6$ & $512$ & $69$ & $0.135$ & $585$ & $100$ & $0.171$ & $84$ & $0.144$ \\
$7$ & $512$ & $139$ & $0.271$ & $585$ & $184$ & $0.315$ & $120$ & $0.205$ \\
\end{tabular}
\caption{Parameters for $m=3$ and $q=8$.}
\end{table}

\begin{table}[!h]
\centering
\begin{tabular}{c|ccc|ccc|cc}
$k$ & $n_A$ & $\dim(A)$ & $R_A$ & $n_P$ & $\dim(P)$ & $R_P$ & $\dim(\PRM)$ & $R_{\PRM}$ \\ \hline
$8$ & $4096$ & $120$ & $0.0293$ & $4369$ & $165$ & $0.0378$ & $165$ & $0.0378$ \\
$9$ & $4096$ & $168$ & $0.041$ & $4369$ & $224$ & $0.0513$ & $220$ & $0.0504$ \\
$10$ & $4096$ & $233$ & $0.0569$ & $4369$ & $302$ & $0.0691$ & $286$ & $0.0655$ \\
$11$ & $4096$ & $320$ & $0.0781$ & $4369$ & $404$ & $0.0925$ & $364$ & $0.0833$ \\
$12$ & $4096$ & $434$ & $0.106$ & $4369$ & $535$ & $0.122$ & $455$ & $0.104$ \\
$13$ & $4096$ & $601$ & $0.147$ & $4369$ & $724$ & $0.166$ & $560$ & $0.128$ \\
$14$ & $4096$ & $854$ & $0.208$ & $4369$ & $1004$ & $0.23$ & $680$ & $0.156$ \\
$15$ & $4096$ & $1377$ & $0.336$ & $4369$ & $1568$ & $0.359$ & $816$ & $0.187$ \\
\end{tabular}
\caption{Parameters for $m=3$ and $q=16$.}
\end{table}

\begin{table}[!h]
\centering
\begin{tabular}{c|ccc|ccc|cc}
$k$ & $n_A$ & $\dim(A)$ & $R_A$ & $n_P$ & $\dim(P)$ & $R_P$ & $\dim(\PRM)$ & $R_{\PRM}$ \\ \hline
$24$ & $32768$ & $3044$ & $0.0929$ & $33825$ & $3405$ & $0.101$ & $2925$ & $0.0865$ \\
$25$ & $32768$ & $3561$ & $0.109$ & $33825$ & $3960$ & $0.117$ & $3276$ & $0.0969$ \\
$26$ & $32768$ & $4192$ & $0.128$ & $33825$ & $4634$ & $0.137$ & $3654$ & $0.108$ \\
$27$ & $32768$ & $4970$ & $0.152$ & $33825$ & $5460$ & $0.161$ & $4060$ & $0.12$ \\
$28$ & $32768$ & $5928$ & $0.181$ & $33825$ & $6471$ & $0.191$ & $4495$ & $0.133$ \\
$29$ & $32768$ & $7250$ & $0.221$ & $33825$ & $7860$ & $0.232$ & $4960$ & $0.147$ \\
$30$ & $32768$ & $9169$ & $0.28$ & $33825$ & $9860$ & $0.292$ & $5456$ & $0.161$ \\
$31$ & $32768$ & $13011$ & $0.397$ & $33825$ & $13824$ & $0.409$ & $5984$ & $0.177$ \\
\end{tabular}
\caption{Parameters for $m=3$ and $q=32$.}
\end{table}

\begin{table}[!h]
\centering
\begin{tabular}{c|ccc|ccc|cc}
$k$ & $n_A$ & $\dim(A)$ & $R_A$ & $n_P$ & $\dim(P)$ & $R_P$ & $\dim(\PRM)$ & $R_{\PRM}$ \\ \hline
$56$ & $262144$ & $44064$ & $0.168$ & $266305$ & $46125$ & $0.173$ & $32509$ & $0.122$ \\
$57$ & $262144$ & $48340$ & $0.184$ & $266305$ & $50520$ & $0.19$ & $34220$ & $0.128$ \\
$58$ & $262144$ & $53401$ & $0.204$ & $266305$ & $55714$ & $0.209$ & $35990$ & $0.135$ \\
$59$ & $262144$ & $59480$ & $0.227$ & $266305$ & $61940$ & $0.233$ & $37820$ & $0.142$ \\
$60$ & $262144$ & $66810$ & $0.255$ & $266305$ & $69431$ & $0.261$ & $39711$ & $0.149$ \\
$61$ & $262144$ & $76717$ & $0.293$ & $266305$ & $79540$ & $0.299$ & $41664$ & $0.156$ \\
$62$ & $262144$ & $90874$ & $0.347$ & $266305$ & $93940$ & $0.353$ & $43680$ & $0.164$ \\
$63$ & $262144$ & $118873$ & $0.453$ & $266305$ & $122304$ & $0.459$ & $45760$ & $0.172$ \\
\end{tabular}
\caption{Parameters for $m=3$ and $q=64$.}
\end{table}

\end{document}